\documentclass[preprint,natbib,nocopyrightspace]{sigplanconf}
\usepackage{amsmath,amssymb,stmaryrd,color,alltt}
\usepackage{mathrsfs} 
\usepackage[usenames,dvipsnames]{xcolor}
\usepackage{amsthm}
\usepackage{bbm}
\usepackage{etex}

\usepackage{graphics}

\usepackage[final,
             bookmarks,
             bookmarksopen,
             colorlinks,
             final,
             linkcolor=blue,
             citecolor=brown,
             pdfstartview=FitH ]%
{hyperref}

\usepackage[all,2cell]{xy}
\UseAllTwocells
\SilentMatrices

\usepackage{bcprules}
\usepackage{proof}
\usepackage{pstring}
\usepackage{paralist}

\newif\iflong\longtrue

\newif\ifdraft\draftfalse

\newif\ifwithappendix\withappendixfalse

\definecolor{DarkGreen}{RGB}{0,100,0}
\definecolor{DarkBlue}{RGB}{0,0,200}
\ifdraft
\newcommand{\tkchanged}[1]{{\textcolor{red}{#1}}}
\newcommand{\lochanged}[1]{{\textcolor{Cyan}{#1}}}
\newcommand{\tk}[1]{{\textcolor{red}{[{#1}---Tsukada]}}}
\newcommand{\lo}[1]{{\textcolor{Cyan}{[{#1}---Ong]}}}
\else
\newcommand{\tkchanged}[1]{{#1}}
\newcommand{\lochanged}[1]{{#1}}
\newcommand{\tk}[1]{}
\newcommand{\lo}[1]{}
\fi

\renewcommand\phi{\varphi}

\newcommand\calP{{\cal P}}
\newcommand\makeset[1]{\{ #1 \}}

\newcommand{\red}{\longrightarrow}
\newcommand{\lam}{\lambda}

\newcommand{\codom}{\mathord{\mathrm{img}}}
\newcommand{\sem}[1]{[\![{#1}]\!]}
\newcommand{\Nat}{\mathbbm{N}}
\newcommand{\PosReal}{\mathbf{R}^+}

\newcommand{\ident}{\mathrm{id}}

\newcommand{\T}{\mathtt{o}}

\newcommand{\op}{\mathit{op}}

\newcommand{\PP}{\mathrm{P}}
\newcommand{\OO}{\mathrm{O}}

\newcommand{\Moves}[1][{}]{\mathcal{M}_{#1}}
\newcommand{\PMoves}[1][{}]{\mathcal{M}_{#1}^\PP}
\newcommand{\OMoves}[1][{}]{\mathcal{M}_{#1}^\OO}

\newcommand{\View}[1]{\lceil{#1}\rceil}
\newcommand{\PView}[1]{\lceil{#1}\rceil}
\newcommand{\OView}[1]{\lfloor{#1}\rfloor}

\newcommand{\CPlay}[1]{\mathbb{P}_{#1}}
\newcommand{\Intr}[1]{\mathit{Int}({#1})}
\newcommand{\CIntr}[1]{\mathbb{I}_{#1}}

\newcommand{\CView}[1]{\mathbb{V}_{#1}}
\newcommand{\ViewEmbed}{\iota}
\newcommand{\CVal}{\mathscr{V}}

\newcommand{\ImmExt}[1]{\mathbf{ie}({#1})}

\newcommand\proj[2]{{{#1} {\restriction_{#2}}}}

\newcommand{\CGame}{\mathbb{G}}
\newcommand{\CDetGame}{\CGame_{\mathit{det}}}
\newcommand{\CWGame}{\CGame_{\mathit{w}}}

\newcommand\cat[1]{\mathbbm{#1}}
\newcommand\opp[1]{{#1}^{\mathit{op}}}
\newcommand\Set{\cat{S}\cat{e}\cat{t}}
\newcommand\cod{{\rm cod}}
\newcommand\Sheaf{\mathbf{Sh}}

\newcommand{\Subst}[3]{{#1} [{#2}] = {#3}}

\newcommand\anglebra[1]{\langle #1 \rangle}

\newcommand{\bag}[1]{\{\!|{#1}|\!\}}

\newcommand{\subty}{\sqsubseteq}
\newcommand{\osubty}{\stackrel{\OO}{\subty}}
\newcommand{\psubty}{\stackrel{\PP}{\subty}}

\newcommand{\eqty}{\approx}

\newcommand{\peqty}{\stackrel{\PP}{\eqty}}

\newcommand{\btt}{{\mathtt{t}\!\!\mathtt{t}}}
\newcommand{\bff}{{\mathtt{f}\!\!\mathtt{f}}}
\newcommand{\bbtt}{{\mathtt{t}\!\mathtt{t}}}
\newcommand{\bbff}{{\mathtt{f}\!\mathtt{f}}}

\theoremstyle{theorem}
\newtheorem{theorem}{Theorem}
\newtheorem{lemma}[theorem]{Lemma}
\newtheorem{corollary}[theorem]{Corollary}
\newtheorem{proposition}[theorem]{Proposition}
\newtheorem*{proposition*}{Proposition}

\theoremstyle{definition}
\newtheorem{definition}[theorem]{Definition}
\newtheorem{example}[theorem]{Example}

\theoremstyle{remark}
\newtheorem{remark}[theorem]{Remark}

\begin{document}
\conferenceinfo{Some Conference 2014}{Jan 1--3, 2014, Oxford, UK}
\copyrightyear{2014}
\copyrightdata{999-9-9999-9999-9}

%
%
\title{Innocent Strategies are Sheaves over Plays}
\subtitle{Deterministic, Non-deterministic and Probabilistic Innocence}

\authorinfo{Takeshi Tsukada}
           {University of Oxford\\ JSPS Postdoctoral Fellow for Research Abroad}
           {tsukada@cs.ox.ac.uk}
\authorinfo{C.-H. Luke Ong}
           {University of Oxford}
           {Luke.Ong@cs.ox.ac.uk}

\maketitle

\begin{abstract}
Although the HO/N games are fully abstract for PCF,
the traditional notion of \emph{innocence} (which underpins these games) is not satisfactory for such language features as non-determinism and probabilistic branching, in that there are stateless terms that are not innocent.  Based on a category of P-visible plays with a notion of embedding as morphisms, we propose a natural generalisation by viewing \emph{innocent strategies as sheaves over (a site of) plays}, echoing a slogan of Hirschowitz and Pous.
Our approach gives rise to fully complete game models in each of the three cases of deterministic, nondeterministic and probabilistic branching. To our knowledge, in the second and third cases, ours are the first such factorisation-free constructions.


\lo{Possible ideas / points to mention in the abstract: The intuition. View functionality.}

\end{abstract}



\section{Introduction}


Game semantics is a powerful paradigm for giving semantics to a variety of programming languages and logical systems. Both HO/N games \cite{HylandO00,Nickau94} (based on arenas and innocent strategies) and AJM games \cite{AbramskyJM00} (based on games equipped with a certain equivalence relation on plays, and history-free strategies) gave rise to the first syntax-independent description of the fully abstract model for the functional programming language PCF. The HO/N-style games, based on arenas and history-sensitive strategies, have been extended to give a fully abstract model for Idealised Algol (PCF extended with locally-scoped references) \cite{AbramskyM97}. Definability, a crucial step of the completeness argument, was established by showing that every compact history-sensitive strategy factorises through an innocent strategy. Using the same factorisation technique, fully abstract HO/N-style game models have been constructed for a spectrum of  Algol-like languages, including Idealised Algol augmented with language features such as 
non-determinism \cite{HarmerM99} and probabilistic branching \cite{DanosH02}. 


Perhaps surprisingly, it is problematic to extend innocent strategies to model PCF extended with non-determinism \cite{Harmer99}. A famous game model by Harmer \cite{Harmer99} is based on factorisation, decomposing a given non-deterministic 
{strategy} into a non-deterministic oracle and a deterministic innocent strategy. To our knowledge, the problem of a factorisation-free fully complete game model for the simply-typed non-deterministic lambda calculus is open; the same problem is also open for lambda calculus augmented with probabilistic branching. This paper presents a new approach to innocent strategies, based on \emph{sheaves over a site of plays}, that yields fully complete game models for lambda calculi extended with these branching constructs.

We are interested in the simply-typed lambda calculi because they have good algorithmic properties, notably, the decidability of \emph{compositional} higher-order model checking \cite{Ong06,TsukadaO14}, which is proved using HO/N-style effect arenas and innocent strategies. Our study of the game semantics of non-deterministic lambda calculus was motivated, in particular, by a desire to introduce abstraction refinement to higher-order model checking based on the \emph{non-deterministic} $\lambda{\bf Y}$-calculus.


Let us begin with a quick overview of the HO/N-style games. Types are interpreted as arenas, and programs of a given type are interpreted as P-strategies for playing in the arena that denotes the type. Recall that an \emph{arena} $A$ is a set of moves $\Moves[A]$ equipped with an \emph{enabling relation}, $({\vdash}_A) \subseteq {(\Moves[A] \cup \makeset{\star}) \times \Moves[A]}$, that gives $A$ the structure of a forest (whereby a move $m$ is a root, called \emph{initial}, just if $\star \vdash_A m$); furthermore, moves on levels $0, 2, 4, \ldots$ of the forest are O-moves, and those that are on levels $1, 3, 5, \ldots$ are P-moves. A \emph{justified sequence} of $A$ is a finite sequence of O/P-alternating moves, $m_1 \, m_2 \, m_3 \ldots m_n$, such that each non-initial move $m_j$ has a pointer to an earlier move $m_i$ (called the \emph{justifier} of $m_i$) such that $m_i \vdash_A m_j$. A key notion of HO/N games is the \emph{view} of a justified sequence: the \emph{P-view} of a justified sequence $s$ is a certain justified subsequence, written $\PView{s}$, consisting of move-occurrences which P considers relevant for determining his next move (similarly for the \emph{O-view} $\OView{s}$ of $s$). A \emph{play} then is a justified sequence, $m_1 \, m_2 \, m_3 \ldots$, that satisfies \emph{Visibility}: for every $i$, if $m_i$ is non-initial then its justifier appears in $\PView{m_1 \, m_2 \ldots m_i}$ (respectively $\OView{m_1 \, m_2 \ldots m_i}$) if $m_i$ is a P-move (respectively O-move). A \emph{strategy} $\sigma$ over an arena $A$ is just a prefix-closed set of even-length plays $s$; $\sigma$ is said to be \emph{deterministic} if whenever $s \, m_1^P, s \, m_2^P \in \sigma$, then $m_1^P = m_2^P$. (We use superscript $P$ to indicate a P-move; similarly for O-move.) Recall that a strategy $\sigma$ is said to be \emph{innocent} if it is \emph{view dependent} i.e.~for all $s \in \sigma$
\begin{equation}
{(s \in \sigma \; \wedge \; \PView{s \, m_1^O \, m_2^P} \in \sigma) \iff s \, m_1^O \, m_2^P \in \sigma}
\label{eq:innocence}
\end{equation}
It is an important property of innocence that---in the sets-of-plays presentation of strategies---every deterministic innocent strategy can be generated by the set of P-views contained in it. The category of arenas and innocent strategies gives rise to a fully complete model of the simply-typed lambda calculus \cite{HylandO00}.

However, as Harmer observed in his thesis \cite{Harmer99}, the notion of innocence breaks down when one tries to use it to model (stateless) non-deterministic functional computation. 

\begin{example}\label{eg:harmer}
Take simply-typed $\lambda$-terms \( \btt := \lam xy. x \) and \( \bff := \lam xy.y \) of type \( \mathbf{B} = \T \to \T \to \T \), and \( M_1 := \lam f. f\, (\btt + \bff) \) and \( M_2 := (\lam f. f\,\btt) + (\lam f. f\,\bff) \) of type \( (\mathbf{B} \to \T) \to \T \), where \( + \) is the construct for non-deterministic branching. Assuming the call-by-name evaluation strategy, these terms can be separated by the term \( N := \lam g. g\,(g\,\bot\,z)\,\bot \), where \( \bot \) is the divergence term, i.e.~\lochanged{\( M_1 \, N \)} may converge but \lochanged{\( M_2 \, N \)} always diverges. In the HO/N game model (see, for example, \cite{Harmer99}), $\sigma_i := \sem{M_i}$ are strategies over the arena \( ((\{ d \} \to \{ d' \} \to \{ c \}) \to \{ b \}) \to \{ a \} \),
for $i = 1, 2$. Note that $\sigma_1$ and $\sigma_2$ are distinct as strategies: for example (we omit pointers from the plays as they can be uniquely reconstructible)
\(
a \, b \, c \, d \, c \, d' \, \in \, (\sigma_1 \setminus \sigma_2).
\)
However $\sigma_1$ and $\sigma_2$ contain the same set of non-empty, even-length P-views, namely, 
\(
\makeset{a \, b, \; a \, b \, c\, d, \; a \, b \, c\, d'}.
\) 
\end{example}

The preceding example shows the sets-of-plays approach works well for expressing, and even composing, non-deterministic 
strategies for stateless programs; the only problem is that, in general, 
the set of P-views cannot be a good generator for these strategies.

The problematic term is \( M_2 \).  It applies the argument \( f \) to \( \btt \) or \( \bff \), non-deterministically, but the branch has already been chosen when \( M_2 \) responds to the initial move.  So \( a\,b\,c\,d\,c\,d' \) is not playable by \( M_2 \), although innocence requires it to.

Our approach is to admit that \( M_2 \) has two possible responses to the initial move: they give the same play \( a\,b \) but have different internal states. Thus a strategy is formally a mapping from plays to sets that represent the internal states.  For example, \( \sem{M_2}(a\,b) = \{ \btt, \bff \} \), where \( \btt \) means the left branch and \( \bff \) the right branch.  Now the P-views for \( \sem{M_2} \) are, say, \( \{ a\,b^{\bbtt},\; a\,b^{\bbtt} c\,d,\; a\,b^{\bbff},\; a\,b^{\bbff} c\,d' \} \). Notice that \( a\,b^{\bbtt} c\,d\,c\,d' \) and \( a\,b^{\bbff} c\,d\,c\,d' \) are no longer forced by innocence to be admissible plays.  From this viewpoint, a deterministic strategy is a mapping from plays to empty or singleton sets. 


In what follows, we discuss how to formalise this idea.

\paragraph{\tkchanged{Ideal-based innocence}}
Before we explain the main ideas behind our sheaf-theoretic approach to innocence, it is helpful to consider a category of plays $\CPlay{A}$, and an alternative view of \emph{deterministic} innocent strategies as \emph{ideals of a preorder presentation} \cite{JungMV08}. The objects of the category $\CPlay{A}$ are (even-length) justified sequences of the arena $A$ satisfying O/P-alternation and P-visibility (but not necessarily O-visibility), which we shall henceforth call \emph{plays} (by abuse of language). The morphisms $f : s \to s'$ are injective maps that preserve moves, justification pointers, and pairs of consecutive O-P moves. A morphism can permute such pairs, provided the pointers are respected. For example, for each play $s$, there are morphisms $\PView{s} \to s$ and $s \to s \, m_1^O \, m_2^P$. 

A \emph{preorder presentation} is a triple $(P, \leq, \triangleright)$ where $(P, \leq)$ is a preorder and ${\triangleright} \subseteq {\calP(P) \times P}$ is called a \emph{covering relation} (we read $U \triangleright s$ as ``$U$ covers $s$''). A subset $I \subseteq P$ is called an \emph{ideal} if 
\begin{inparaenum}
\item[(I1)] $I$ is lower-closed i.e.~if $t \in I$ and $s \leq t$ then $s \in I$, and
\item[(I2)] for every covering $U \triangleright s$, if $U \subseteq I$ then $s \in I$. 
\end{inparaenum}
A preorder presentation can be extracted from the category $\CPlay{A}$, namely, $(\mathit{Obj}(\CPlay{A}), \leq, \triangleright)$ whereby $s \leq s'$ just if there is a morphism $f : s \to s'$, and $U \triangleright s$ just if $U = \makeset{s_\xi}_{\xi \in \Xi}$ for some family of morphisms, $\makeset{f_\xi : s_\xi \to s}_{\xi \in \Xi}$, which is \emph{jointly surjective}, meaning that the union of the set of move-occurrences that appear in the image of $f_\xi$, as $\xi$ ranges over $\Xi$, is the set of all move-ocurrences in the play $s$.

Then ideals of the preorder presentation $(\mathit{Obj}(\CPlay{A}), \leq, \triangleright)$ are innocent strategies. Notice that, because $s \leq s \, m_1^O \, m_2^P$ and $\PView{s \, m_1^O \, m_2^P} \leq s \, m_1^O \, m_2^P$, condition (I1) of ideal gives the $\Leftarrow$-direction of (\ref{eq:innocence}). Further since the set $\makeset{s, \PView{s \, m_1^O \, m_2^P}}$ covers $s \, m_1^O \, m_2^P$, condition (I2) gives the other direction of (\ref{eq:innocence}).



\paragraph{\tkchanged{From ideals to sheaves}}
A presheaf, $F : \cat{C}^\op \to \Set$, is a contravariant functor, assigning data (\lochanged{a set of ``internal states''}) to each object $s$ of $\cat{C}$. The definition of sheaf of a site is technical, and a version is presented in the preliminaries subsection. Here we can think of a sheaf over a site as an extension of the notion of an ideal of a preorder presentation. A \emph{site} is a pair $(\cat{C}, J)$ where $\cat{C}$ is a category, and $J$, called a \emph{coverage}, assigns to each object $s$ of $\cat{C}$ a collection of \emph{covering families}, each of the form $\makeset{f_\xi : s_\xi \to s}_{\xi \in \Xi}$. Intuitively a presheaf, $F : \cat{C}^\op \to \Set$, is a sheaf over the site $(\cat{C}, J)$ just if the data assigned to a given object $s$ (meaning the elements of $F(s)$) can be systematically tracked by the data \emph{locally defined} over the family $\makeset{f_\xi : s_\xi \to s}_{\xi \in \Xi}$ (meaning the elements of $F(s_\xi)$, as $\xi$ ranges over $\Xi$), for all covering families of $s$; further, every \emph{matching family} of such locally assigned data uniquely determines a datum assigned to $s$ (an element of $F(s)$). Thus, take the site $(\CPlay{A}, J)$ where $J(s)$ consists of the jointly surjective families of morphisms with codomain $s$, then $(\mathit{Obj}(\CPlay{A}), \leq, \triangleright)$ is a preorder presentation, as discussed in the preceding. In our sheaf-theoretic approach, an innocent strategy of arena $A$, whether deterministic or not, \emph{is} a sheaf $\sigma$ over the site $(\CPlay{A}, J)$. The intuition is that a sheaf $\sigma : \CPlay{A}^\op \to \Set$ that maps every $s$ to either a singleton set or the emptyset (which is so if the strategy $\sigma$ is deterministic) corresponds to an ideal $I_\sigma$ of the associated preorder presentation whereby $s \in I_\sigma$ if and only if $\sigma(s) \neq \emptyset$.

\paragraph{Our contributions}

Our thesis is that sheaves $\CPlay{A}^{\mathit{op}} \to \Set$ generalise innocent strategies of the arena $A$.  (Indeed the sheaves approach seems more general than innocence, since it appears capable of capturing the computation of single-threaded (history-sensitive) strategies as well.)

Given arenas $A, B$ and $C$, we define a category $\CIntr{A, B, C}$ whose objects are \emph{interaction sequences} of the triple $(A, B, C)$ in the usual sense, and whose morphisms $f : u \to u'$ are injective maps that preserve moves, justification pointers, and \emph{basic blocks} 
(which are sequences of moves that begin with an O-move of $A \Rightarrow C$, and end with a P-move of $A \Rightarrow C$, with all intermediate moves from $B$). Let $u \in \CIntr{A, B, C}$, we write $\proj{u}{A, B}, \proj{u}{B, C}$ and $\proj{u}{A, C}$ for the standard projections of $u$ to the component arenas. Given sheaves $\sigma_1 : \CPlay{A, B}^\op \to \Set$ and $\sigma_2 : \CPlay{B, C}^\op \to \Set$, there is a natural way to compose them. (We write $\CPlay{A, B}$ to mean $\CPlay{A \Rightarrow B}$.) Define a presheaf $\sigma_1;\sigma_2 : \CPlay{A, C}^\op \to \Set$, which acts on objects as follows:
\[
(\sigma_1 ; \sigma_2)(s) := \coprod_{u \in \CIntr{A, B, C} : \proj{u}{A, C} = s} \sigma_1(\proj{u}{A, B}) \times \sigma_1(\proj{u}{A, B})
\] We show that the composite $\sigma_1 ; \sigma_2$ is well-defined:
\begin{enumerate}[(i)] 
\item $\sigma_1 ; \sigma_2 : \CPlay{A, C}^{\op} \to \Set$ is a sheaf
\item $\sigma_1 ; \sigma_2$ is the left Kan extension of the 
\tkchanged{functor} $F : \CIntr{A, B, C}^\op \to \Set$, whose action on objects is $u \mapsto \sigma_1(\proj{u}{A, B}) \times \sigma_2(\proj{u}{B, C})$, along the projection functor $\CIntr{A, B, C}^{\op} \to \CPlay{A, C}^{\op} $. 
\item composition is associative up to natural isomorphism: 
\[
(\sigma_1 ; \sigma_2) ; \sigma_3 \cong \sigma_1 ; (\sigma_2 ; \sigma_3)
\]
\end{enumerate}
Furthermore, the category whose objects are arenas and whose morphisms $\sigma : A \to B$ are (equivalence classes of isomorphic) sheaves $\sigma : \CPlay{A, B}^\op \to \Set$ is cartesian closed. 

Just as innocent strategies are view dependent, so there is a compelling sense in which sheaves on plays, $\CPlay{A, B}^\op \to \Set$, depend on (indeed, are determined by) sheaves on \emph{views},  $\CView{A, B}^\op \to \Set$, where $\CView{A, B}$ is a full subcategory of $\CPlay{A, B}$. The subcategory $\CView{A, B}$, whose objects are \emph{nonempty P-views}, is a preorder, and the induced topology is trivial (every object has a unique covering sieve which is maximal). Since every object in $\CPlay{A, B}$ has a covering sieve by objects of the subcategory $\CView{A, B}$, thanks to the Comparison Lemma \cite{Beilinson12,Verdier72}, $\iota^\ast : \Sheaf(\CPlay{A, B}) \to \Sheaf(\CView{A, B})$ gives an equivalence of the respective categories (of sheaves), where $\iota : \CView{A, B} \hookrightarrow \CPlay{A, B}$ is the embedding.

Sheaves on views are important because they are easier to understand and calculate with than sheaves on plays; conversely, composition of the latter is easier to describe than that of the former. Let $\tau_M : \CView{A}^\op \to \Set$ be the denotation of a non-deterministic $\lambda$-term $M$. Then given $p \in \CView{A}$, $\tau_M(p)$ corresponds to the set of all possible runs (\emph{qua} plays) of $M$ whose P-view is $p$. Returning to Example~\ref{eg:harmer}:

\begin{example}\label{eg:harmer2}
Using the notation in Example~\ref{eg:harmer}, let \( p_0 = a\,b \), \( p_{1} = a\,b\,c\,d \) and \( p_{2} = a\,b\,c\,d' \). For $i = 1, 2$ define \( \tau_i \in \Sheaf(\CView{(\{ d \} \to \{ d' \} \to \{ c \}) \to \{ b \}, \{ a \})}) \) to be the sheaf-over-views denotation of $M_i$. Then
\[
\begin{array}{lcl}
  \tau_1(p_0) = \{ x_1 \} 
& \quad &
    \tau_2(p_0) = \{ x_{21}, x_{22} \}
  \\
  \tau_1(p_1) = \{ y_1 \} 
  & \quad &
    \tau_2(p_1) = \{ y_{21} \hspace{18pt} \} 
  \\
  \tau_1(p_2) = \{ z_1 \} 
& \quad &
    \tau_2(p_2) = \{ \hspace{18pt} z_{22} \}

\end{array}
\]
Notice that in the set of plays $\sem{M_2}$, there are two independent plays (which have the P-view) $p_0$.
\end{example}

Our approach gives rise to fully complete game models in each of the three cases of deterministic, nondeterministic and probabilistic branching. To our knowledge, in the second and third cases, ours are the first such factorisation-free constructions.


%
%
%

\paragraph{Related work}

The standard notion of innocence does not work well for certain language features, such as non-determinism. To address the deficiency, Levy \cite{Levy13} proposed a category of P-visible plays and \emph{viewing morphisms}. This is essentially our category $\CPlay{A}$ of plays. However in \emph{op.~cit.~}an innocent strategy $\sigma$ is still defined to be a certain \emph{set of plays}, namely, a lower-closed set of objects of the category: if $t \in \sigma$ and $s \to t$ is a morphism, then $s \in \sigma$. Because this definition captures only one of the two requirements of innocence (i.e.~$\Leftarrow$ of (\ref{eq:innocence})), Levy's construction will likely not yield accurate (fully complete) models of the non-deterministic $\lambda$-calculus.

A related approach by Hirschowitz et al.~\cite{HirschowitzP12,EberhartHS13} does view strategies as presheaves (and sheaves) on a category of plays. However, in contrast to our focus on higher-type computation, they are concerned with CCS-style concurrent computation which they model as multi-player games. Strategies are presheaves on a category of plays ${\mathbb E}_X$ over a \emph{position} $X$, and a strategy is deemed innocent if it is determined by its restriction to a subcategory of views ${\mathbb V}_X \hookrightarrow {\mathbb E}_X$. A \emph{position} is an undirected graph describing the channel-based communication topology connecting the players, and plays are certain ``glueings'' of moves over a position, with moves built-up using CCS constructs. Thus the connexions with our work seem superficial. 

Winskel et al.~\cite{StatonW10,RideauW11} have worked extensively on causal games as models of true concurrency, from the viewpoint of strategies as event structures with symmetries. Recently Clairambault et al.~\cite{CastellanCW14} built a conservative extension of HO/N games in a truly concurrent framework. An extensional quotient of their model yields a fully abstract model of PCF with parallel or.

Perhaps surprisingly, the question of what is the proper notion of \emph{innocence in the presence of non-determinism} is still open. Harmer and McCusker \cite{HarmerM99} seem only concerned with \emph{stateful} non-deterministic programs, namely non-deterministic Idealised Algol.

\paragraph{Technical preliminaries}


In the following we review the basic definitions of coverage, Grothendieck topology and sheaves, and refer the reader to the book \cite{MacLaneM92} for an exposition.

A \emph{coverage} on a category $\mathbb{C}$ is a map $J$ assigning to each object $s$ of $\mathbb{C}$ a collection $J(s)$ of families $\makeset{f_\xi : s_\xi \to s}_{\xi \in \Xi}$ of maps with codomain $s$, called \emph{covering families}, such that the system of families is ``stable under pullback'', meaning:
if $\makeset{f_\xi : s_\xi \to s}_{\xi \in \Xi}$ is a covering family and $g : t \to s$ is a map, then there is a covering family, $\makeset{h_\nu : t_\nu \to t}_{\nu \in N}$, such that each $g \circ h_\nu$ factors through some $f_\xi$.
A number of \emph{saturation conditions} are often imposed on a coverage for convenience. A \emph{site} is a category $\cat C$ equipped with a coverage $J$, written $(\cat{C}, J)$.

Given a family $S = \makeset{f_\xi : s_\xi \to s}_{\xi \in \Xi}$ of maps with codomain $s$, and a presheaf $F : \opp{\cat{C}} \to \Set$, a family of elements $\makeset{x_\xi \in F(s_\xi)}_{\xi \in \Xi}$ is said to be \emph{matching for $S$} if for all maps $g : t \to s_\xi$ and $h : t \to s_{\xi'}$, if $f_\xi \circ g = f_{\xi'} \circ h$ then $F(g)(x_\xi) = F(h)(x_{\xi'})$. An \emph{amalgamation} for the family $\makeset{x_\xi \in F(s_\xi)}_{\xi \in \Xi}$ is an $x \in F(s)$ such that $F(f_\xi)(x) = x_\xi$ for every $\xi \in \Xi$. A presheaf $F : \opp{\cat{C}} \to \Set$ is a \emph{sheaf} for a family $S = \makeset{f_\xi : s_\xi \to s}_{\xi \in \Xi}$ of maps just if every matching family for $S$ has a unique amalgamation. A presheaf is a \emph{sheaf} for a site if it is a sheaf for every covering family of the site. 

A \emph{sieve} on an object $s$ in a category $\cat{C}$ is a family of maps with codomain $s$ that are closed under precomposition with maps in $\cat{C}$. Given a family $\makeset{f_\xi : s_\xi \to s}_{\xi \in \Xi}$, the {sieve} it generates is the family of all maps $g : t \to s$ with codomain $s$ that factor through some $f_\xi$. A presheaf is a sheaf for a family $\makeset{f_\xi : s_\xi \to s}_{\xi \in \Xi}$ if, and only if, it is a sheaf for the sieve it generates. If $S$ is a sieve on $s$ and $g : t \to s$ is a map, we define $g^\ast(S)$ to be the sieve on $t$ consisting of all maps $h$ with codomain $t$ such that $g \circ h$ factors through some map in $S$.

A \emph{Grothendieck topology} is a map $J$ that assigns to each object $s$ of $\cat{C}$ a collection $J(s)$ of sieves on $s$, called \emph{covering sieves}, that satisfies the following:
\begin{enumerate}[(i)]
\item The maximal sieve, $\makeset{f \mid \cod(f) = s}$, is in $J(s)$.

\item (\emph{Stability}) If $S \in J(s)$ then $h^\ast(S) \in J(t)$ for every map $h : t \to s$.

\item (\emph{Transitivity}) If $S \in J(s)$ and $R$ is a sieve on $s$ such that $h^\ast(R) \in J(t)$ for every $h : t \to s$ in $S$, then $R \in J(s)$.
\end{enumerate}

\begin{lemma}
For every coverage, there is a unique Grothendieck topology that has the same sheaves.
\end{lemma}

\paragraph{Notation}
We write \( \Nat \) for the set of all positive integers.  For an integer \( n \), we define \( [n] := \{ k \mid 1 \le k \le n \} \) and \( [n]_0 := \{ k \mid 0 \le k \le n \} \). For a category \( \mathbb{C} \), we write \( x \in \mathbb{C} \) to mean that \( x \) is an object of \( \mathbb{C} \).


\section{Sites of Plays}
This section defines sites of plays over an arena. The innocent strategies are just sheaves over those sites.  The category of plays has a subcategory of views.  We prove that the sheaves over plays is equivalent to sheaves over views: this generalises view dependency to non-deterministic computation.

\subsection{Plays}

The definition of arenas is standard (as in \cite{HylandO00}) except that all moves are questions.
\begin{definition}[Arena]
An \emph{arena} is a tuple \( A = (\Moves[A], \lambda_A, {\vdash_A}) \), where \( \Moves[A] \) is a finite set of \emph{moves}, \( \lambda_A : \Moves[A] \to \{ \PP, \OO \} \) is an ownership function and \( (\vdash_A) \subseteq (\{ \star \} + \Moves[A]) \times \Moves[A] \) is an \emph{enabling relation} that satisfies the following conditions:
\begin{inparaenum}[$(1)$]
\item for every \( m \in \Moves[A] \), there is a unique \( x \in \{ \star \} + \Moves[A] \) such that \( x \vdash_A m \), and
\item if \( \star \vdash_A m \), then \( \lambda_A(m) = \OO \).  If \( m \vdash_A m' \), then \( \lambda_A(m) \neq \lambda_A(m') \).
\end{inparaenum}
\end{definition}

For an arena \( A \), the set \( \OMoves[A] \) of \emph{O-moves} is defined as \( \{ m \in \Moves[A] \mid \lambda_A(m) = \OO \} \).  The set of \emph{P-moves} is defined by \( \PMoves[A] := \{ m \in \Moves[A] \mid \lambda(m) = \PP \} \).  A move \( m \) is \emph{initial} if \( \star \vdash_A m \).  {An arena is \emph{prime} if it has exactly one initial move.}

We write \( \{ m \} \) for the arena that has one O-move \( m \) and no P-moves.  For a prime arena \( A \) and an arena \( B \), \( B \to A \) is the arena whose moves are \( \Moves[A] + \Moves[B] \) where the initial \( B \)-move is enabled by the unique initial \( A \)-move.  For example, \( \{ m_1 \} \to \{ m_2 \} \to \{ m_3 \} \) consists of an O-move \( m_3 \) and P-moves \( m_1 \) and \( m_2 \) with \( \star \vdash m_3 \), \( m_3 \vdash m_1 \) and \( m_3 \vdash m_2 \).

Unlike the standard formalisation, in which notions such as justified sequences and plays are parametrised by arenas, we parametrise them by a pair of arenas \( (A, B) \), corresponding to the exponential arena \( A \Rightarrow B \) in the standard formalisation.  This change simplifies some definitions.
\begin{definition}[Arena pair]
Let \( A = (\Moves[A], \lambda_A, {\vdash_A}) \) and \( B = (\Moves[B], \lambda_B, {\vdash_B}) \) be arenas.  The \emph{moves of \( (A, B) \)} is the disjoint union of moves, say \( \Moves[A, B] := \Moves[A] + \Moves[B] \).  We define \emph{P-moves} by \( \PMoves[A, B] := \OMoves[A] + \PMoves[B] \) and \emph{O-moves} by \( \OMoves[A, B] := \PMoves[A] + \OMoves[B] \).  For \( m, m' \in \Moves[A, B] \), we write \( m \vdash_{A, B} m' \) just if either
\begin{inparaenum}[(1)]
\item \( m, m' \in \Moves[A] \) and \( m \vdash_{A} m' \), or
\item \( m, m' \in \Moves[B] \) and \( m \vdash_{B} m' \), or
\item \( \star \vdash_{B} m \in \Moves[B] \) and \( \star \vdash_{A} m' \in \Moves[A] \).
\end{inparaenum}
We write \( \star \vdash_{A, B} m \) just if \( \star \vdash_{B} m \in \Moves[B] \).
\end{definition}

For a pair \( (A, B) \), an \emph{initial \( A \)-move} is a move \( m \in \Moves[A] \subseteq \Moves[A, B] \) such that \( \star \vdash_A m \): do not confuse it with \( \star \vdash_{A, B} m \), which is impossible.  An \emph{initial \( B \)-move} is defined similarly.

\begin{definition}[Justified sequence]
Let \( (A, B) \) be a pair of arenas.  A \emph{justified sequence of \( (A, B) \)} is a finite sequence of moves equipped with justification pointers.  Formally it is a pair of functions \( s : [n] \to \Moves[A, B] \) and \( \varphi : [n] \to [n]_0 \) (for some \( n \)) such that
\begin{itemize}
\item \( \varphi(k) < k \) for every \( k \in [n] \), and
\item \( \varphi \) respects the enabling relation: \( \varphi(k) \ne 0 \) implies \( s(\varphi(k)) \vdash_{A, B} s(k) \), and \( \varphi(k) = 0 \) implies \( \star \vdash_{A, B} s(k) \).
\end{itemize}
As usual, by abuse of notation, we often write \( m_1 \, ßm_2 \dots m_n \) for a justified sequence such that \( s(i) = m_i \) for every \( i \), leaving the justification pointers implicit.  Further we use \( m \) and \( m_i \) as metavariables of \emph{occurrences} of moves in justified sequences.  We write \( m_i \curvearrowleft m_j \) if \( \varphi(j) = i > 0 \) and \( \star \curvearrowleft m_j \) if \( \varphi(j) = 0 \).  We call \( m_i \) the \emph{justifier of \( m_j \)} when \( m_i \curvearrowleft m_j \).  We write \( \curvearrowleft^+ \) for the transitive closure of \( \curvearrowleft \).  We write \( |s| \) for the length of \( s \).
\end{definition}

It is convenient to relax the domain \( [n] \) of justified sequences to arbitrary linearly-ordered finite sets such as a subset of \( [n] \).  For example, given a justified sequence \( (s: [n] \to \Moves[A,B], \varphi : [n] \to [n]_0) \), consider a subset \( I \subseteq [n] \) that respects the justification pointers, i.e.~\( k \in I \) implies \( \varphi(k) \in I \cup \{ 0 \} \).  Then the restriction \( (s{\upharpoonright_I} : I \to \Moves[A,B], \varphi{\upharpoonright_I} : I \to \{ 0 \} \cup I) \) is a justified sequence in the relaxed sense.  Through the unique monotone bijection \( \alpha : I \to [n'] \), we identify the restriction with the justified sequence in the narrow sense. 

A justified sequence is \emph{alternating} if \( s(k) \in \OMoves[A,B] \) iff \( k \) is odd (so \( s(k) \in \PMoves[A, B] \) iff \( k \) is even).

\begin{definition}[P-View/P-visibility]
Let \( m_1 \dots m_n \) be an alternating justified sequence over \( (A, B) \).  Its \emph{P-view} \( \View{m_1 \dots m_n} \) (or simply \emph{view}) is a subsequence defined inductively by:
\begin{align*}
  \View{m_1 \dots m_n} &:= \View{m_1 \dots m_{n-1}}\,m_n &\textrm{(if \( m_n \in \PMoves[A, B] \))} \\
  \View{m_1 \dots m_n} &:= m_n &\textrm{(if \( \star \curvearrowleft m_{n} \in \OMoves[A, B] \))} \\
  \View{m_1 \dots m_n} &:= \View{m_1 \dots m_k} \, m_n &\textrm{(if \( m_k \curvearrowleft m_n \in \OMoves[A, B] \)).}
\end{align*}
More formally, given an alternating justified sequence \( s \) of length \( n \), its view is a subset \( I \subseteq [n] \).  The above equation gives the restriction of \( s \) to \( I \).  A view is, in general, not a justified sequence since the justifier of a move may have been removed.

Let \( m_k \) be a P-move in the sequence.  Its justifier is said to be \emph{P-visible} if it is in \( \View{m_1 \dots m_k} \).  An alternating justified sequence is \emph{P-visible} 
if the justifier of each P-move occurrence in \( s \) is P-visible.
\end{definition}

\begin{definition}[Play]
An alternating justified sequence over a pair \( (A, B) \) of arenas is a \emph{play} just if it is P-visible and its last move is a P-move \( m \in \PMoves[A, B] \).
\end{definition}

\begin{remark}
In contrast to the standard definition of play in innocent game semantics (as in \cite{HylandO00}), we do not require \emph{O-visibility}.  This is technically convenient because O-visibility is not preserved by commutations (see Definition~\ref{def:play-commute}).  Note also that a play may have several initial moves, i.e.~we do not assume \emph{well-openness}.
\end{remark}

\subsection{Morphisms between plays that respects P-views}
In the traditional HO/N game models, the set of plays are considered as a poset ordered by the prefix ordering.  In this subsection, we introduce a richer structure to plays, 
\lochanged{organising them into a category. This is essentially the category introduced by Levy~\cite{Levy13}.}

It is useful to view an even-length alternating justified sequence is a sequence of \emph{pairs} of O- and P-moves, which we shall call a \emph{block} (or an \emph{O-P block}).

\begin{definition}[Morphism between plays]\label{def:play-morphism}
Let \( m_1 \dots m_n \) and \( m'_1 \dots m'_{n'} \) be plays of length \( n \) and \( n' \), respectively.  A \emph{morphism between plays} is an injection \( f : [n] \to [n'] \) s.t.~for every \( k \in [n] \)
\begin{enumerate}[(i)]
\item \( m_k = m'_{f(k)} \) (as moves),
\item \( m_i \curvearrowleft m_k \) implies \( m'_{f(i)} \curvearrowleft m'_{f(k)} \) (and similarly for \( \star \curvearrowleft m_k \)), and
\item if an O-move \( m_k \) is followed by a P-move \( m_{k+1} \), then \( m'_{f(k)} \) is followed by \( m'_{f(k+1)} \) (i.e.~\( f(2l-1) + 1 = f(2l) \) for all \( l \)). 
\end{enumerate}
I.e.~a morphism between plays is an injective map between {O-P blocks} 
that preserves moves and justification pointers. 
We define \( \codom(f) := \{ f(k) \mid k \in [n] \} \subseteq [n'] \).
\end{definition}

\begin{example}\label{eg:morphisms}
\begin{inparaenum}[$(i)$]
\item Let \( s = m_1 \dots m_n \) be a play and \( s' = m_1 \dots m_l \) be its (even-length) prefix.  Then \( f(i) = i \) (for \( i \le l \)) is a morphism \( f : s' \to s \).  In other words, each prefix \( s' \le s \) induces a morphism.  (But this may not be the unique morphism of \( s' \to s \).)
\item Let \( s = m_1 \dots m_{n-1} m_n \) and assume that \( m_k \curvearrowleft m_{n-1}^{\OO} \).  Then we have \( f : (m_1 \dots m_k m_{n-1} m_n) \to s \), where \( f(i) = i \) (if \( i \le k \)), \( f(k+1) = n-1 \) and \( f(k+2) = n \).
\item For every play \( s \), we have a unique morphism \( f : \PView{s} \to s \) that maps the last move of \( \PView{s} \) to the last move of \( s \) (though there may exist another morphism that does not satisfy this condition).  In this sense, the morphisms of the category is an generalisation of the notion of P-views.
\item Let \( s = s_0 \, m_{n-3} m_{n-2} \, m_{n-1} m_n \) be a play and assume that the justifier of \( m_{n-1} \) is not \( m_{n-2} \).  Let \( s' \) be the play \( s_0 \, m_{n-1} m_{n} \, m_{n-3} m_{n-2} \) obtained from \( s \) by commuting O-P blocks \( m_{n-3} m_{n-2} \) and \( m_{n-1} m_n \).  There is an isomorphism \( f : s \to s' \), given by \( f(i) = i \) (if \( i < n-3 \)), \( f(n-3) = n-1 \), \( f(n-2) = n \), \( f(n-1) = n-3 \) and \( f(n) = n-2 \).
\end{inparaenum}
\end{example}

\tkchanged{
\begin{example}
Let \( (A, B) = ((\{ d \} \to \{ c \}) \to \{ b \},\; \{ a \}) \) be a pair of arenas and the play \( s = a\,b\,c\,d\,c'\,d' \) (where \( c = c' \) and \( d = d' \) as moves) over \( (A, B) \) in which \( c' \) points to \( b \) and all other moves are justified by their preceding move.  Let \( s' = a\,b\,c\,d \) be another play of \( (A, B) \).  Then \( s' \) can be regarded as a prefix of \( s \) and as the P-view \( \PView{s} \) of \( s \).  The first perspective induces the morphism \( f : s' \to s \), where \( f(i) = i \) (for \( i \in [4] \)), and the second perspective does \( g : s' \to s \), where \( g(1) = 1 \), \( g(2) = 2 \), \( g(3) = 5 \) and \( g(4) = 6 \).
\end{example}
}

\begin{definition}[Category of plays]
Let \( A \) and \( B \) be arenas.  The \emph{category \( \CPlay{A, B} \) of plays} has plays of \( (A, B) \) as objects and as morphisms those defined above.
\end{definition}

\tkchanged{
\begin{lemma}
\( \CPlay{A, B} \) has pullbacks.
\end{lemma}
\begin{proof}
Let \( f : s_1 \to t \) and \( g : s_2 \to t \).  They are injective maps \( f : [|s_1|] \to [|t|] \) and \( g : [|s_2|] \to [|t|] \).  Let \( I = \codom(f) \cap \codom(g) \).  The restriction of \( t \) to \( I \) is the pullback \( s_1 \times_t s_2 \).
\end{proof}
}

\tkchanged{We give another definition of morphisms via \emph{commutation}.}
\begin{definition}[Commutation of non-interfering blocks]\label{def:play-commute}
Let \( s \) be an even-length alternating justified sequence over \( (A, B) \).  Let \( m_1 m_1'\, m_2 m_2' \) be an adjacent pair of O-P \tkchanged{blocks} in \( s \), i.e.~\( s = t \, m_1 m_1'\, m_2 m_2' \, t' \), where \( m_1 \) and \( m_2 \) are O-moves.  We say that the pairs are \emph{non-interfering} if the justifier of \( m_2 \) is \emph{not} \( m_1' \).  The \emph{commuted sequence} \( s' \) is defined by \( s' := t\, m_2 m_2' \, m_1 m_1'\, t' \) (in which the justification pointers are modified accordingly).
\end{definition}

A commuted sequence is not always a justified sequence: if \( m_2' \) is justified by \( m_1 \), then \( m_2' \) in the commuted sequence is not well-justified.  \tkchanged{If the justified sequence is P-visible, the commuted sequence is a justified sequence.}
Furthermore the converse also holds.
\begin{lemma}\label{lem:exchange-imply-visibility}
Let \( P \) be a set of even-length alternating justified sequences over \( (A, B) \).  Suppose that \( P \) is closed under commutations, i.e.~for every sequence \( s \in P \) and every non-interfering adjacent pairs of blocks in \( s \), the commuted sequence is also in \( P \).  Then all justified sequences in \( P \) are plays.
\end{lemma}
\begin{proof}
Let \( s = m_1 \dots m_n \in P \) and \( m_k \) be a P-move occurrence in \( s \).  We prove that the justifier of \( m_k \) is in the P-view \( \PView{m_1 \dots m_k} \).  By commuting pairs as much as required, we can reach a sequence, say \( s' = m'_1 \dots m'_n \), such that \( m'_l \) is the move corresponding to \( m_k \) in \( s \) and \( m'_{2i} \curvearrowleft m'_{2i+1} \) for every \( 2i < l \).  This means that \( \PView{m'_1 \dots m'_l} = m'_1 \dots m'_l \) and hence the justifier of \( m'_l \) is in the view.  Since P-visibility is preserved and reflected by the commutation of non-interfering blocks, 
\( s \) is P-visible.
\end{proof}

\tkchanged{Every morphism can be expressed as the prefix embedding followed by commutations.  This is insightful and technically useful.}



\begin{lemma}\label{lem:map-play-char}
Every \( f : s \to t \) in \( \CPlay{A, B} \) can be decomposed as
\[
  s \xrightarrow{f} t \;=\; s \xrightarrow{\le} t_0 \xrightarrow{g_1} t_1 \xrightarrow{g_2} \dots \xrightarrow{g_n} t_n,
\]
where \( n \ge 0 \), \( t_n = t \) and \( g_i \) is a commutation of adjacent O-P blocks in \( t_{i-1} \) for every \( i \in [n] \).  (This decomposition is not unique.)
\end{lemma}
\begin{proof}
Let \( f : s \to t \) and \( t = m_1 \dots m_n \).  If \( f \) is induced by the prefix, then we complete the proof.  Otherwise, there is an odd number \( k \le |s| \) such that either \( f(k) - 2 \notin \codom(f) \) or \( f(l) = f(k) - 2 \) for some \( l > k \).  Then we claim that \( m_{f(k)-2} m_{f(k)-1} \) and \( m_{f(k)} m_{f(k)+1} \) in \( t \) is a non-interfering pair.  Suppose otherwise, i.e.~the justifier of \( m_{f(k)} \) is \( m_{f(k)-1} \).  Then \( f(k) - 1 \in \codom(f) \) since \( f \) preserves the justification pointer.  Let \( l' \le |s| \) be the index such that \( f(l') = f(k) - 1 \).  Since \( s(l') \curvearrowleft s(k) \), we have \( l' < k \).  Because \( l' \) is even, we have \( f(l'-1) = f(l')-1 = f(k)-2 \).  In summary, we have \( l < k \) such that \( f(l) = f(k) - 2 \), that contradict the assumption.  So the adjacent O-P blocks \( m_{f(k)-2} m_{f(k)-1} \) and \( m_{f(k)} m_{f(k)+1} \) in \( t \) is non-interfering.

Consider the commutation \( h : t \to t' \) and the inverse \( h^{-1} : t' \to t \), which is also a commutation.  By applying the same argument to \( h \circ f : s \to t' \), \( h \circ f \) can be decomposed as \( g_n \circ \dots \circ g_1 \circ g_0 \), where \( g_0 \) is induced by the prefix and \( g_i \) (\( i > 1 \)) is a commutation.  This inductive argument is justified by the same way as the termination of the bubble sort.  Then \( f = h^{-1} \circ g_n \circ \dots \circ g_1 \circ g_0 \).
\end{proof}

\begin{remark}
\tkchanged{Let \( \sigma \) be an innocent strategy in the standard sense, i.e.~an even-prefix closed subset of plays with a certain condition.  Then \( s \in \sigma \) and \( f : s' \to s \) in \( \CPlay{A, B} \) implies \( s' \in \sigma \).  To see this, observe that a commutation of \( s \in \sigma \) is in \( \sigma \) and use Lemma~\ref{lem:map-play-char}.}
\end{remark}

\subsection{Topology of \( \CPlay{A, B} \)}
As for the innocent strategies \( \sigma \) for deterministic calculi, which is a set of plays, a play \( s = m_1 \dots m_k \) is in the strategy \( \sigma \) iff P-views for (even-)prefixes are in \( \sigma \), i.e.~\( \{ \PView{m_1 \dots m_k} \mid k = 2, 4, \dots n \} \subseteq \sigma \).  We use the Grothendieck topology to capture this condition.

\begin{definition}[Covering family / sieve]
A family of morphisms \( \{ f_{\xi} : s_{\xi} \to s \}_{\xi \in \Xi} \) is said to \emph{cover} \( s \) when they are jointly surjective, i.e.\ \( \bigcup_{\xi \in \Xi} \codom(f_{\xi}) = [n] \), where \( n \) is the length of \( s \).  A \emph{covering sieve} is a sieve that is a covering family.  
\tkchanged{By abuse of notation, we write \( \CPlay{A, B} \) for the site associated with this topology.}
\end{definition}

\tkchanged{
\begin{example}\label{eg:covering}
\begin{inparaenum}[$(i)$]
\item For a play \( s = m_1 \dots m_{n} \), the family \( \{ f: (m_1 \dots m_{n-2}) \to s,\; g: \PView{s} \to s \} \) is a covering family.  Here \( f \) is induced by the prefix and \( g \) by the P-view (see Example~\ref{eg:morphisms}).
\item For a play \( s = m_1 \dots m_n \), the family \( \{ f_k : \PView{m_1 \dots m_k} \to s \}_{k \in \{ 2, 4, \dots, n \}} \) is a covering family.  Here \( f_k \) is the composite of the P-view embedding and the prefix embedding, i.e.,
\[
  \PView{m_1 \dots m_k} \stackrel{f_k}{\longrightarrow} s
  \;=\;
  \PView{m_1 \dots m_k} \longrightarrow (m_1 \dots m_k) \longrightarrow s.
\]
The covering family generalises \tk{resembles? extends?} the set of P-views of the prefixes.
\item The covering family is finer than the set of P-views.  Let \( s = m_1 m_2 m_1' m_2' \) (the repetition of \( m_1 m_2 \) twice).  Then \( \{ f : m_1 m_2 \to s \} \), where \( f(1) = 1 \) and \( f(2) = 2 \), is \emph{not} a covering family.  However \( \{ f : m_1 m_2 \to s,\ g : m_1 m_2 \to s \} \), where \( g(1) = 3 \) and \( g(2) = 4 \), is a covering family.  Notice that those two families have the same set of the domain, say \( \{ m_1 m_2 \} \), which is the set of P-views of \( s \).  
\end{inparaenum}
\end{example}
}

\begin{definition}
An \emph{innocent strategy} is a sheaf over \( \CPlay{A, B} \).
\end{definition}

\begin{remark}
Let \( \sigma \) be a functor \( \CPlay{A, B}^{\op} \to \Set \).  It is \emph{pre-deterministic} if \( \sigma(s) \) is empty or singleton for every \( s \).  A pre-deterministic functor can be determined by the set \( P_{\sigma} = \{ s \in \CPlay{A, B} \mid \sigma(s) \neq \emptyset \} \).  Since \( \sigma \) is a functor, the set \( P_{\sigma} \) is lower closed, i.e.~\( s \in P_{\sigma} \) and \( f : s' \to s \) in \( \CPlay{A, B} \) implies \( s' \in \sigma \).  A pre-deterministic functor \( \sigma \) is a sheaf just if \( s = m_1 \dots m_n \in P_{\sigma} \) iff \( \{ \PView{m_1 \dots m_k} \mid k = 2, 4, \dots, n \} \subseteq P_{\sigma} \).  To see this, observe that \( \{ \PView{m_1 \dots m_k} \to s \mid k = 2, 4, \dots, n \} \) is a covering family and the family of unique elements \( \{ x_k \in \sigma(\PView{m_1 \dots m_k}) \}_{k} \) is a matching family and thus there is an amalgamation \( x \in \sigma(s) \).  In this sense, for pre-deterministic strategies, the innocence is equivalent to the sheaf condition.  However, if \( \sigma(s) \) may have more than one element, innocence based on the set of views differs from the sheaf condition.
\end{remark}

\subsection{\tkchanged{Sheaves over \( \CPlay{A, B} \) and its restriction to P-views}}
In innocent game models for deterministic calculi (such as \cite{HylandO00}), one often considers the restriction of strategies to P-views.  A remarkable property is that an innocent strategy (\emph{qua} set of plays) is completely determined by the subset of P-views it contains.  After all, innocence means \emph{view dependence}. 

In this subsection, we shall see that a similar property holds for sheaves over plays \( \CPlay{A, B} \).  This property comes from the topological structure of plays: every play is covered by P-views \tkchanged{(see Example~\ref{eg:covering}$(ii)$).}
This observation gives a justification of defining innocent strategies as sheaves.

\begin{definition}[Subcategory of P-views]
A play \( s \in \CPlay{A, B} \) is a \emph{P-view} if \( \PView{s} = s \) and \tkchanged{\emph{\( s \) is not empty}}.  We use \( p \) as a metavariable ranging over P-views.  The \emph{category of P-views} \( \CView{A, B} \) is the full subcategory of \( \CPlay{A, B} \) consisting of P-views.  We write \( \ViewEmbed : \CView{A, B} \hookrightarrow \CPlay{A, B} \) for the embedding.  Henceforth we fix the topology for \( \CView{A, B} \) to be that induced%
\footnote{
\tk{Possible footnote}  Given a site \( \mathbb{C} \) and a full subcategory \( \mathbb{D} \hookrightarrow \mathbb{C} \), the induced topology on \( \mathbb{D} \) is defined by: a sieve \( S \) on \( \mathbb{D} \) is covering iff the sieve \( (S) := \{ f \circ h \mid f \in S,\; \mathrm{dom}(f) = \mathrm{codom}(h) \} \) on \( \mathbb{C} \) generated from \( S \) is covering.
}
from \( \CPlay{A, B} \): \tkchanged{it is the trivial topology, i.e.~every P-view has only one covering sieve, namely, the maximal sieve.}
\end{definition}

\tkchanged{The category of P-views is a poset.  We write \( (p' \le p) \) and \( (p \ge p') \) for the unique morphism \( f : p' \to p \) (if it exists).}

\tkchanged{Because the topology is trivial, a sheaf over \( \CView{A, B} \) is just a functor \( \CView{A, B}^{\op} \to \Set \).}
A sheaf \( \sigma \in \Sheaf(\CPlay{A, B}) \) 
induces a sheaf \( \sigma \circ \ViewEmbed \) over \( \CView{A, B} \).  The strategy \( \sigma \) can be reconstructed from the restriction to P-views \( \sigma \circ \ViewEmbed \) (up to natural isomorphism).

\begin{lemma}[Comparison]\label{lem:comparison}
The functor \( \iota^* : \Sheaf(\CPlay{A, B}) \ni \sigma \mapsto \sigma \circ \iota \in \Sheaf(\CView{A, B}) \) induces an equivalence of categories.
\end{lemma}

Since every play has a covering by P-views, Lemma~\ref{lem:comparison} follows from a standard result, known as the Comparison Lemma \cite{Verdier72} (see, for example, \cite[Prop.~p.~721]{Beilinson12} which generalises the classical result in SGA4).
However an explicit description of the adjoint \( \iota_* : \Sheaf(\CView{A, B}) \to \Sheaf(\CPlay{A, B}) \) is insightful and worth clarifying.

Let \( \tau \in \Sheaf(\CView{A, B}) \) be a sheaf over P-views.  \tkchanged{Let \( s = m_1 \dots m_n \) be a non-empty play and \( p_k := \PView{m_1 \dots m_k} \) for every even number \( k \).}  We define a set of \tkchanged{\emph{\( \tau \)-annotations}} for \( s \): a \tkchanged{\( \tau \)-annotation} is a sequence \( e_2 e_4 \dots e_n \), where \( e_k \in \tau(p_k) \) for every even number \( k \), subject to the following condition: for every even number \( k \le n \), if \( m_l^{\PP} \curvearrowleft m_{k-1}^{\OO} \), then \( e_l = \tau(p_l \le p_k)(e_k) \).  For a non-empty play \( s \in \CPlay{A, B} \), we write \( (\iota_* \tau)(s) \) for the set of all \tkchanged{\( \tau \)-annotations}.

Given \( f : s \to s' \), which is an injective map \tkchanged{\( f : [|s|] \to [|s'|] \),} 
the morphism \( (\iota_* \tau)(f) : (\iota_* \tau)(s') \to (\iota_* \tau)(s) \) is defined by:
\[
  (\iota_* \tau)(f) : e_2 e_4 \dots e_{|s'|} \mapsto e_{f(2)} e_{f(4)} \dots e_{f(|s|)}.
\]
\tkchanged{We define \( (\iota_* \tau)(\varepsilon) := \{ \ast \} \) for the empty sequence.  Then \( \iota_* \tau : \CPlay{A, B}^{\op} \to \Set \) is a functor.}

\begin{example}\label{eg:view-dependency}
Consider an arena pair \( ((\{ d \} \to \{ d' \} \to \{ c \}) \to \{ b \},\; \{ a \}) \) and let \( p_0 = a\,b \), \( p_{1} = a\,b\,c\,d \) and \( p_{2} = a\,b\,c\,d' \) (in which every move is justified by its predecessor).  Define \( \tau_1, \tau_2 \in \Sheaf(\CView{(\{ d \} \to \{ d' \} \to \{ c \}) \to \{ b \},\; \{ a \}}) \) as follows:
\[
\begin{array}{llcll}
  \tau_1(p_0) &= \{ x_1 \} & &
  \tau_2(p_0) &= \{ x_{21}, x_{22} \}
  \\
  \tau_1(p_1) &= \{ y_1 \} & &
  \tau_2(p_1) &= \{ y_{21} \hspace{18pt} \} 
  \\
  \tau_1(p_2) &= \{ z_1 \} & &
  \tau_2(p_2) &= \{ \hspace{18pt} z_{22} \}
  \\[3pt]
  \tau_1(f)(y_1) &= x_1 & &
  \tau_2(f)(y_{21}) &= x_{21}
  \\
  \tau_1(g)(z_1) &= x_1 & &
  \tau_2(g)(z_{22}) &= x_{22},
\end{array}
\]
where \( f : (a\,b) \to (a\,b\,c\,d) \) and \( g : (a\,b) \to (a\,b\,c\,d') \).  Then
\[
  (\iota_* \tau_1)(a\,b\,c\,d\,c\,d') = \{ x_1 y_1 z_1 \}
\qquad
  (\iota_* \tau_2)(a\,b\,c\,d\,c\,d') = \{ \; \}.
\]
We write \( P_{\sigma} := \{ s \in \CPlay{A, B} \mid \sigma(s) \neq \emptyset \} \) and \( V_{\sigma} := \{ p \in \CView{A, B} \mid \sigma(p) \neq \emptyset \} \).  Then \( V_{\sigma_1} = V_{\sigma_2} \) but \( P_{\sigma_1} \neq P_{\sigma_2} \).  The set-of-views approach fails to distinguish \( \sigma_1 \) from \( \sigma_2 \).
\end{example}

\begin{proposition}
\( \iota_* \tau \in \Sheaf(\CPlay{A, B}) \) for every \( \tau \in \Sheaf(\CView{A, B}) \).
\end{proposition}
\iflong 
\begin{proof}
Let \( S = \{ f_\xi : s_{\xi} \to s \}_{\xi \in \Xi} \) be a covering sieve and \( \{ x_\xi \in (\iota_* \tau)(s_{\xi}) \}_{\xi \in \Xi} \) be a matching family.  Each \( x_\xi \) is a \(\tau\)-annotation \( e_{\xi, 2} e_{\xi, 4} \dots e_{\xi, |s_{\xi}|} \).  It suffices to give an annotation \( e_2 e_4 \dots e_n \) for \( s \) (here \( n = |s| \)).  Let \( k \le n \) be an even number.  Since \( S \) is a covering sieve, it must be jointly surjective, i.e.~\( k \in \codom(f_{\xi}) \) for some \( \xi \).  When \( f_{\xi}(l_k) = k \), we define \( e_k = e_{\xi, l_k} \).  This does not depend on the choice of \( \xi \) since \( x_{\xi} \) is a matching family.  The resulting sequence \( e_2 \dots e_n \) satisfies the required conditions.  The uniqueness is trivial.
\end{proof}

\begin{proposition}
\( \iota_* \) and \( \iota^* \) form an adjoint equivalence. 
\end{proposition}
\begin{proof}
Let \( \tau \in \Sheaf(\CView{A, B}) \).  For a P-view \( p = m_1 \dots m_n \), an annotation \( a_2 a_4 \dots a_n \in (\iota_* \tau)(p) \) is uniquely determined by \( a_n \), since \( a_k = \tau(f_k)(a_n) \) for the unique \( f_k : (m_1 \dots m_k) \to (m_1 \dots m_n) \).  This gives a bijection \( \psi_p \) for each \( p \) from \( \tau(p) \) to \( (\iota_* \tau)(p) \), and to \( (\iota^* \iota_* \tau)(p) \) through \( (\iota^* \iota_* \tau)(p) = (\iota_* \tau)(p) \).

For the other direction, let \( \sigma \in \Sheaf(\CPlay{A, B}) \).  Let \( s = m_1 \dots m_n \) be a play.  Then \( x \in (\iota_* \iota^* \tau)(s) \) is a sequence \( e_2 e_4 \dots e_n \) such that, for every even number \( k \le n \), \( e_k \in \sigma(\PView{m_1 \dots m_k}) \) and \( e_l = \sigma(f_k)(e_k) \) if \( m_l \curvearrowleft m_{k-1} \).  This means that \( \{ a_k \}_{k \in \{ 2, 4, \dots, n \}} \) is a matching family of \( \{ \PView{m_1 \dots m_k} \to s \}_{k \in \{ 2, 4, \dots, n \}} \).  Since \( \sigma \) is a sheaf, there exists a bijection \( \phi_s \) from \( (\iota_* \iota^* \tau)(s) \) to \( \tau(s) \).

It is easy to see that \( (\iota_*, \iota^*, \psi, \phi) \) is an adjanction.
\end{proof}
\fi 


\section{Interaction and composition}
This section introduces the notion of \emph{interaction sequences} and defines the composition \( (\sigma_1; \sigma_2) \in \Sheaf(\CPlay{A, C}) \) of sheaves \( \sigma_1 \in \Sheaf(\CPlay{A, B}) \) and \( \sigma_2 \in \Sheaf(\CPlay{B, C}) \), generalising the composition of deterministic innocent strategies as in \cite{HylandO00}.  The composition is associative up to isomorphism, and the arenas and sheaves form a CCC (where isomorphic sheaves are identified).

\subsection{Interaction sequences}
\begin{definition}[Justified sequence]
Let \( (A, B, C) \) be a triple of arenas.  The enabling relation \( {\vdash_{A,B,C}} \) for the triple is defined by:
\begin{itemize}
\item For \( X \in \{ A, B, C \} \), if \( m \vdash_X m' \), then \( m \vdash_{A,B,C} m' \).
\item If \( \star \vdash_C m \in \Moves[C] \), then \( \star \vdash_{A, B, C} m \).
\item If \( \star \vdash_C m \in \Moves[C] \) and \( \star \vdash_B m' \in \Moves[B] \), then \( m \vdash_{A,B,C} m' \).
\item If \( \star \vdash_B m \in \Moves[B] \) and \( \star \vdash_A m' \in \Moves[A] \), then \( m \vdash_{A,B,C} m' \).
\end{itemize}
A \emph{justified sequence} of the triple is a sequence over \( \Moves[A] + \Moves[B] + \Moves[C] \) equipped with justification pointers that respect the enabling relation \( {\vdash_{A,B,C}} \).
\end{definition}


A justified sequence \( s \) of a triple \( (A, B, C) \) induces justified sequences of \( (A, B) \), \( (B, C) \) and \( (A, C) \), basically by the restriction of moves.  The \emph{projection to the component \( (B, C) \)}, written \( \proj{s}{B, C} \), is just the restriction.  The \emph{projection to the component \( (A, B) \)}, written \( \proj{s}{A, B} \), is the restriction to moves in \( \Moves[A, B] \) in which \( \star \curvearrowleft m \) for an initial \( B \)-move \( m \) (whereas \( m' \curvearrowleft m \) in \( s \) for an initial \( C \)-move \( m' \)).  The \emph{projection to the component \( (A, C) \)}, written \( \proj{s}{A, C} \), is the restriction to moves in \( \Moves[A, C] \) in which an initial \( A \)-move \( m \) is justified by the move \( m' \) such that \( m' \curvearrowleft m'' \curvearrowleft m \) (so \( m'' \) is an initial \( B \)-move and \( m' \) an initial \( C \)-move).

\begin{definition}[Interaction sequence]
Let \( (A, B, C) \) be a triple of arenas.  A justified sequence \( s \) over \( (A, B, C) \) is an \emph{interaction sequence} if
\begin{itemize}
\item The last move is in \( \PMoves[A, C] = \OMoves[A] + \PMoves[C] \), and
\item \( \proj{s}{A, B} \) and \( \proj{s}{B, C} \) are plays of \( (A, B) \) and \( (B, C) \), respectively.
\end{itemize}
\end{definition}

\paragraph{Switching condition and basic blocks}
Before defining the morphisms between interaction sequences, we introduces a useful tool to analyse the interaction sequences.

\begin{definition}[Switching condition]
Let \( (A, B, C) \) be a triple of arenas.  A sequence over \( \Moves[A] + \Moves[B] + \Moves[C] \) is said to satisfy the \emph{switching condition} if it is accepted by the following automaton with the initial state \( OOO \) of which all states are accepting.
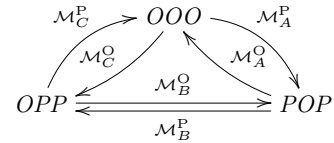
\[\xymatrix@R-0.1cm@C.8cm{
& \mathit{OOO} \ar@/^1pc/[dr]^{\PMoves[A]} \ar@/^/[dl]_{\OMoves[C]} & \\
\mathit{OPP} \ar@/^1pc/[ur]^{\PMoves[C]} \ar@<1pt>[rr]^{\OMoves[B]} & &
\mathit{POP} \ar@<2pt>[ll]^{\PMoves[B]} \ar@/^/[ul]_{\OMoves[A]}
}
\]
\tkchanged{A state express the owners of the next moves for components \( (A, B) \), \( (B, C) \) and \( (A, C) \) in this order.}
\tk{NOTE: The elements of a triple like \( OOO \) mean the expected moves for components \( (A, B) \), \( (B, C) \) and \( (A, C) \) in this order.  For example, since a move \( m \in \PMoves[A] \) is an O-move both in \( (A, B) \) and in \( (A, C) \), this move can appear when the current state is \( O?O \) (but \( OPO \) never happens), moving to the state \( P?P \).}
\tk{What would be the best reference for this notion?  I found the above diagram in ``Notes on game semantics'' by Curien.} \lo{This mechanism is known to game semanticists. It is probably best to cite Curien -- I don't know of another reference.}
\end{definition}
The switching condition generalises the O/P-alternation of justified sequences for a pair \( (A, B) \).

\begin{lemma}\label{lem:interaction-satisfy-switching}
Interaction sequences satisfy the switching condition.
\end{lemma}
\iflong 
\begin{proof}
Observe that each state of the automaton is determined by the first two component.  Thus the O-P alternation for \( (A, B) \) and \( (B, C) \) components suffice for the switching condition.
\end{proof}
\fi 

Recall that basic constituents of plays are pairs of consecutive O-P move occurrences, called O-P blocks.
\tkchanged{
Thanks to the switching condition (Lemma~\ref{lem:interaction-satisfy-switching}), we know that interaction sequences consist of what we shall call \emph{basic blocks}: a basic block is a sequence of consecutive move occurrences in the interaction sequence, starting from a move in \( \PMoves[A, C] \) and ending with a move in \( \OMoves[A, C] \), possibly having moves in \( \Moves[B] \) as intermediate moves.
}
\tk{A similar notion: see Hyland and Ong (2000) Proposition 5.4.} 

\paragraph{The category of interaction sequences}
Given a triple \( (A, B, C) \), a \emph{generalised P-move} is a move in \( \OMoves[A] + \Moves[B] + \PMoves[C] \).  This can be written as \( \PMoves[A, C] + \Moves[B] \) and as \( \PMoves[A, B] + \PMoves[B, C] \).  An \emph{generalised O-move} is a move in \( \PMoves[A] + \Moves[B] + \OMoves[C] \).

\begin{definition}\label{def:interact-seq-morphism}
Let \( (A, B, C) \) be a triple of arenas and \( s, s' \) be interaction sequences over \( (A, B, C) \).  Suppose that \( s = m_1 \dots m_n \) and \( s' = m'_1 \dots m'_{n'} \).  A \emph{morphism between \( s \) and \( s' \)} is an injective map \( f : [n] \to [n'] \) which satisfies:
\begin{itemize}
\item \( m_k = m'_{f(k)} \) (as moves),
\item \( m_i \curvearrowleft m_k \) implies \( m'_{f(i)} \curvearrowleft m'_{f(k)} \) (and similarly for \( \star \curvearrowleft m_k \)), and
\item if a generalised O-move \( m_k \) is followed by \( m_{k+1} \), \( m'_{f(k)} \) is followed by \( m'_{f(k+1)} \) \tkchanged{(i.e.~\( f(k+1) = f(k)+1 \))}.
\end{itemize}
In other words, a morphism between interaction sequences is an injective map between the respective occurrence-sets that preserve moves, justification pointers and basic blocks.
\end{definition}

\begin{definition}
Given arenas \( A \), \( B \) and \( C \), the \emph{category of interaction sequences}, written as \( \CIntr{A, B, C} \), has interaction sequences 
as objects and morphisms defined above.
\end{definition}

\begin{remark}
One can introduce the topology to \( \CIntr{A, B, C} \) as follows, though we shall not use them:  A family of morphisms \( \{ f_\xi : s_\xi \to s \}_{\xi \in \Xi} \) in \( \CIntr{A, B, C} \) is said to \emph{cover \( s \)} if they are jointly surjective, i.e.~\( \bigcup_{\xi \in \Xi} \codom(f_{\xi}) = [n] \), where \( n \) is the length of \( s \).
\end{remark}

\paragraph{Projection to \( (A, C) \) component}
The projections of an interaction sequence onto \( (A, B) \) and \( (B, C) \) components are plays by definition.  We show that the projection onto \( (A, C) \) component is also a play.

\begin{definition}[Commuting an adjacent pair of non-interfering blocks]
Let \( u \) be an interaction sequence of \( (A, B, C) \).  Let \[ m_1 v_1 m_1' \, m_2 v_2 m_2' \] be an adjacent pair of basic blocks in \( u \), where \( m_1 \) and \( m_2 \) are moves in \( \OMoves[A, C] \), \( m_1' \) and \( m_2' \) are moves in \( \PMoves[A, C] \), and \( v_1 \) and \( v_2 \) are sequences of moves in \( \Moves[B] \); i.e.~\( u = u_0 \, m_1 v_1 m_1' \, m_2 v_2 m_2'\, u_1 \).  We say that the pair of basic blocks are \emph{non-interfering} if the justifier of \( m_2 \) is \emph{not} \( m_1' \).  The \emph{commuted sequence} \( u' \) is defined by \( u' := u_0 \,m_2 v_2 m_2' \, m_1 v_1 m_1' \, u_1 \) (in which the justification pointers are modified accordingly).
\end{definition}

\begin{lemma}\label{lem:commute-int}
Let \( u \) be an interaction sequence of \( (A, B, C) \) and let \( v \) be obtained from \( u \) by commuting an adjacent pair of non-interfering blocks. Then \( v \) is an interaction sequence.
\end{lemma}
\iflong 
\begin{proof}
Let \( u = s'\ t_1 \ t_2 \ s'' \) and \( v = s' \ t_2 \ t_1 \ s'' \),  where \( t_1 \) and \( t_2 \) are non-interfering basic blocks, i.e.~the justifier of the first move in \( t_2 \) is not the last move in \( t_1 \).  Let \( t_2 = m_1 \dots m_k \).  We prove tho following claim:
\begin{quote}
  Let \( m_i \) be a move in \( t_2 \).  Then the justifier of \( m_i \) is not in \( t_1 \).
\end{quote}
We prove this by induction on \( i \).

We prove the base case \( i = 1 \).  Since \( m_1 \in \PMoves[A, C] \), by the definition of the basic block, its justifier is in \( \OMoves[A, C] \).  Because \( t_1 \) is a basic block, the unique move in \( \OMoves[A, C] \) is the last move.  By the assumption the justifier of \( m_1 \) differs from the last move of \( t_1 \), as desired.

We prove the induction step.  Let \( m_i \) be a move in \( t_2 \) (\( i > 1 \)).  Then \( m_i \) is either in \( \PMoves[B, C] \) or in \( \PMoves[A, B] \).  Suppose that \( m_i \in \PMoves[B, C] \).  Since \( u \) is an interaction sequence, \( \proj{u}{B, C} \) is a play.  In particular the justifier of \( m_i \) is in \( \PView{\proj{(s'\,t_1\,m_1\,\dots\,m_i)}{B, C}} \).  Let \( n_1 \dots n_l \) be the P-view.  We show that no move in this sequence is in \( t_1 \).  First \( n_l = m_i \) and its immediate predecessor \( n_{l-1} \) are in \( t_2 \).  The preceding move \( n_{l-2} \) is pointed by \( n_{l-1} \), so by the induction hypothesis, \( n_{l_2} \) is not in \( t_1 \).  If \( n_{l_2} \) is in \( s' \), then all preceding moves are in \( s' \).  If \( n_{l-2} \) is in \( t_2 \), by iterating the same argument, we conclude that \( n_1 \dots n_l \) does not contain moves in \( t_1 \).  Since \( \proj{u}{B, C} \) is a play, its justifier is in its P-view.  Hence not a move in \( t_1 \).

We prove that \( \proj{v}{B, C} \) is a play, using the above claim.  Notice that \( \proj{v}{B, C} \) is obtained by commuting adjacent O-P blocks in \( \proj{u}{B, C} \) as much as required.  The above claim implies that every O-P block in \( \proj{t_1}{B, C} \) does not interfere to any O-P block in \( \proj{t_2}{B, C} \).  Since commutation of non-interfering O-P blocks preserves P-visibility, \( \proj{v}{B, C} \) is a play.  Similarly \( \proj{v}{A, B} \) is a play.
\end{proof}
\fi 

\begin{lemma}
For every interaction sequence \( u \) of \( (A, B, C) \), the projection \( \proj{u}{A, C} \) is a play.
\end{lemma}
\begin{proof}
Let \( u \) be an interaction sequence of \( (A, B, C) \).  We define the set \( P \) of interaction sequences as the least set that satisfies
\begin{inparaenum}[$(1)$]
\item \( u \in P \), and
\item if \( v \in P \) and \( v' \) is obtained from \( v \) by commuting a non-interfering basic blocks, then \( v' \in P \).
\end{inparaenum}
In \((2)\), \( v' \) is an interaction sequence by Lemma~\ref{lem:commute-int}.
Consider \( \proj{P}{A, C} := \{ \proj{v}{A, C} \mid v \in P \} \).  This is a set of alternating justified sequences of \( (A, C) \) that is closed under the commutations.  By Lemma~\ref{lem:exchange-imply-visibility}, each element in \( \proj{P}{A, C} \) is a play.  So \( \proj{u}{A, C} \) is a play.
\end{proof}



\paragraph{Projections as functors}
Given an interaction sequence \( u \in \CIntr{A, B, C} \), the projections \( \proj{u}{A, B} \), \( \proj{u}{B, C} \) and \( \proj{u}{A, C} \) are plays of \( (A, B) \), \( (B, C) \) and \( (A, C) \), respectively.  Those projections are naturally extended to functors: given interaction sequences \( u, v \in \CIntr{A, B, C} \) and a morphism \( f : u \to v \), the restriction \( \proj{f}{A, B} \) of \( f \) is a morphism \( \proj{f}{A, B} : \proj{u}{A, B} \to \proj{u}{A, B} \).

\begin{lemma}
The projection \( \proj{}{A, B} : \CIntr{A, B, C} \to \CPlay{A, B} \), \( \proj{}{B, C} : \CIntr{A, B, C} \to \CPlay{B, C} \) and \( \proj{}{A, C} : \CIntr{A, B, C} \to \CPlay{A, C} \) are functors.
\end{lemma}
\iflong 
\begin{proof}
Recall that \( \proj{u}{A, B} \) is the restriction of \( u \) to \( I_{A, B}^u := \{ i \in [|u|] \mid u(i) \in \Moves[A, B] \} \).  A morphism \( f : u \to u' \) in \( \CIntr{A, B, C} \), which is an injection \( f : [|u|] \to [|u'|] \) on sets, is mapped to \( f{\upharpoonright_{I_{A, B}^u}} : I_{A, B}^u \to I_{A, B}^{u'} \).  It is easy to see that this is functorial.
\end{proof}
\fi 



\begin{lemma}\label{lem:fibration}
Let \( f : s \to t \) in \( \CPlay{A, C} \) and \( v \in \CIntr{A, B, C} \) such that \( \proj{v}{A, C} = t \).  Then there exists unique \( \bar{f}_{v} : u \to v \) in \( \CIntr{A, B, C} \) such that \( \proj{\bar{f}_{v}}{A, C} = f \).
\end{lemma}

\begin{proof}
\tkchanged{Observe that the O-P blocks in \( t \) bijectively correspond to the basic blocks in \( v \).  Since a morphism \( f : s \to t \) is an injective map between O-P blocks, the bijection between O-P blocks and basic blocks determines \( \bar{f}_{v} : u \to v \).  So \( \bar{f}_v \) is unique if it exists.  We prove the existence.  If \( f \) is a commutation, Lemma~\ref{lem:commute-int} suffices.  If \( f \) is an embedding induced by a prefix, existence of \( \bar{f}_v \) is trivial.  Lemma~\ref{lem:map-play-char} says that these cases are enough to prove the claim.}
\end{proof}

In other words, \( {\proj{}{A, C}} : \CIntr{A, B, C} \to \CPlay{A, C} \) is a fibration of which each fibre is a discrete category.  We write \( f^*(v) \) for the object \( u \) in the lemma \tkchanged{and \( \bar{f}_v \) for the morphism.}


\subsection{Composition}\label{sec:composition}
\tkchanged{
Let \( \sigma_1 \in \Sheaf(\CPlay{A, B}) \) and \( \sigma_2 \in \Sheaf(\CPlay{B, C}) \) be sheaves.  We define the composite \( (\sigma_1; \sigma_2) : \CPlay{A, C}^{\op} \to \Set \), which shall be proved to be a sheaf.  For a play \( s \in \CPlay{A, C} \), the set \( (\sigma_1; \sigma_2)(s) \) is defined by
\[
  (\sigma_1; \sigma_2)(s) := \coprod_{{u \in \CIntr{A, B, C}} :\  {\proj{u}{A, C} = s}} \sigma_1(\proj{u}{A, B}) \times \sigma_2(\proj{u}{B, C}).
\]
}
\tkchanged{So an element in \( (\sigma_1; \sigma_2)(s) \) is represented by a triple \( (u, e_1, e_2) \), where \( u \in \CIntr{A, B, C} \) such that \( \proj{u}{A, C} = s \), \( e_1 \in \sigma_1(\proj{u}{A, B}) \) and \( e_2 \in \sigma_2(\proj{u}{B, C}) \).  For a morphism \( f : s \to t \) in \( \CPlay{A, C} \), \( (\sigma_1; \sigma_2)(f) \) is a function given by
\[
  (u, e_1, e_2) \mapsto (f^*(u),\;\; \sigma_1(\proj{\bar{f}_u}{A, B})(e_1),\;\; \sigma_2(\proj{\bar{f}_u}{B, C})(e_2)).
\]
\lochanged{In the preceding, we use the common notation \( x \cdot f \) to mean \( F(f)(x) \) where \( F : \mathbb{C}^{\op} \to \Set \), \( f : s \to t \) is a morphism of \( \cat{C} \), and \( x \in F(t) \).}  By this notation, the second component can be written as \( e_1 \cdot (\proj{\bar{f}_u}{A, B}) \) and the third component as \( e_2 \cdot (\proj{\bar{f}_u}{B, C}) \).
}

\tkchanged{Categorically, the composite is the left Kan extension.}

\begin{lemma}
Assume \( \sigma_1 \in \Sheaf(\CPlay{A, B}) \) and \( \sigma_2 \in \Sheaf(\CPlay{B, C}) \).  Let \( F : \CIntr{A, B, C}^{\op} \to \Set \) be a functor defined by \( F(u) := \sigma_1(\proj{u}{A, B}) \times \sigma_2(\proj{u}{B, C}) \).  Then the composite \( (\sigma_1; \sigma_2) \) is the left Kan extension of \( F \) along the projection \( \pi : \CIntr{A, B, C}^{\op} \to \CPlay{A, C}^{\op} \).
\[\xymatrix@R-0.1cm@C.8cm{
   \CPlay{A, C}^{\op} \ar[rrrd]^{\sigma_1; \sigma_2} \\
\CIntr{A, B, C}^{\op} \ar[r]^-{} \ar[u]^-{\pi}
&
\CPlay{A, B}^{\op} \times \CPlay{B, C}^{\op}  \ar[r]^-{\sigma_1 \times \sigma_2}
&
\Set \times \Set \ar[r]
&
\Set
}
\]
\end{lemma}
\begin{proof}
\tkchanged{
The universal natural transformation \( \alpha : F \to (\sigma_1; \sigma_2) \circ \pi \) is given by
\[
  \alpha_{u} : F(u) \ni (e_1, e_2) \mapsto (u, e_1, e_2) \in (\sigma_1; \sigma_2)(\pi(u)).
\]
Assume a functor \( H : \CPlay{A, C}^{\op} \to \Set \) and a natural transformation \( \beta : F \to H \circ \pi \).  Thus for every \( u \in \CIntr{A, C, B} \), we have
\(
  \beta_{u} : F(u) \to H(\pi(u)).
\)
Now \( \gamma_s : (\sigma_1; \sigma_2)(s) \to H(s) \) is defined by
\[
  \gamma_s(u, e_1, e_2) := \beta_u(e_1, e_2)
\]
(recall that \( (\sigma_1; \sigma_2)(s) = \coprod_{u:\; \pi(u) = s} \sigma_1(\proj{u}{A, B}) \times \sigma_2(\proj{u}{B, C}) \)).
Then \( \gamma \) is natural and \( \gamma_{\pi(u)} \circ \alpha_{u} = \beta_u \) for all \( u \).  Uniqueness of \( \gamma \) comes from the universal property of coproducts.
}
\end{proof}

\begin{remark}
In the traditional set-theoretic HO/N game semantics, the composite of strategies \( P_{A, B} \) and \( P_{B, C} \) (i.e.~ even-prefix closed subsets of plays over \( (A, B) \) and over \( (B, C) \), respectively) is defined by \( (P_{A, B}; P_{B, C}) := \{ s \in \CPlay{A, C} \mid \exists u \in \CIntr{A, B, C}.\ \proj{u}{A, B} \in P_{A, B} \textrm{ and } \proj{u}{B, C} \in P_{B, C} \} \).  Our composition satisfies \( (P_{\sigma_1}); (P_{\sigma_2}) = P_{(\sigma_1; \sigma_2)} \), where \( P_{\sigma} = \{ s \mid \sigma(s) \neq \emptyset \} \).
\end{remark}

The composite of sheaves is again a sheaf.
\begin{theorem}
Let \( \sigma_1 \in \Sheaf(\CPlay{A, B}) \) and \( \sigma_2 \in \Sheaf(\CPlay{B, C}) \) be sheaves.  Then \( \sigma_1; \sigma_2 \) is a sheaf over \( \CPlay{A, C} \).
\end{theorem}
\begin{proof}
Let \( s = m_1 \dots m_n \in \CPlay{A, C} \) be a play, \( \{ f : s_f \to s \}_{f \in S} \in J(s) \) be a covering sieve and \( \{ x_{f} \in (\sigma_1; \sigma_2)(s_{f}) \}_{f \in S} \) be a matching family.  By the definition of \( \sigma_1; \sigma_2 \), we have
\[
  x_{f} = (u_{f}, y_{f}, z_{f}) \in \coprod_{u} \sigma_1(\proj{u}{A, B}) \times \sigma_2(\proj{u}{B, C}).
\]

We claim that there exists \( u \) such that:
\begin{itemize}
\item \( \proj{u}{A, C} = s \), and
\item \( u_{f} = f^*(u) \) for every \( f \in S \).
\end{itemize}
If such \( u \) exists, there is a bijective correspondence between basic blocks of \( u \) and O-P blocks of \( s \).  This correspondence tells us the start and the last moves of each block.  So it suffices to fill the intermediate \( B \)-moves for each basic block.  Consider the \( k \)th basic block.  Since \( S \) is a covering sieve, we have a morphism \( f : s_f \to s \in S \) such that \( 2k \in \codom(f) \) (recall that \( k \)th O-P block is \( m_{2k-1} m_{2k} \)).  Let \( l \) be the index such that \( f(l) = 2k \).  Recall that \( x_f = (u_f, y_f, z_f) \) with \( \proj{u_f}{A, C} = s_f \).  Then the basic block of \( u_f \) corresponding to the O-P block \( m'_{l-1} m'_l \) in \( s_f = m'_1 \dots m'_{|s_f|} \) tells us the \( k \)th basic block of \( u \).  This is independent of the choice of \( f \) since \( \{ x_f \}_{f \in S} \) is a matching family.  Now by the construction, \( u_f = f^*(u) \).

Then we have a family \( T := \{ \bar{f}_u : f^*(u) \to u \}_{f \in S} \).  This family is jointly surjective, i.e.~\( \bigcup_{f \in S} \codom(\bar{f}_u) = [|u|] \), since \( S \) is jointly surjective on O-P blocks of \( s \), which bijectively correspond to basic blocks of \( u \).  Hence \( \proj{T}{A, B} := \{ \proj{\bar{f}_u}{A, B} \mid f \in S \} \) and \( \proj{T}{B, C} := \{ \proj{\bar{f}_u}{B, C} \mid f \in S \} \) are covering families and \( \{ y_f \}_{f \in S} \) and \( \{ z_f \}_{f \in S} \) are matching families of them.  Hence there exist amalgamations \( x \in \sigma_1(\proj{u}{A, B}) \) and \( y \in \sigma_2(\proj{u}{B, C}) \).  Then \( (u, x, y) \in (
\sigma_1; \sigma_2)(s) \) is the amalgamation.

The uniqueness of \( u \) follows from the construction and the amalgamations \( x \) and \( y \) are unique since \( \sigma_1 \) and \( \sigma_2 \) are sheaves.
\end{proof}

\subsection{Associativity}
The associativity of composition (up to natural isomorphism) is proved by studying ``generalised'' interaction sequences \( \CIntr{A, B, C, D} \) that have two internal components.  This is a standard technique.

\begin{definition}
Given a quadruple \( (A, B, C, D) \) of arenas, the enabling relation \( \vdash_{A, B, C, D} \) on \( \Moves[A, B, C, D] := \Moves[A] + \Moves[B] + \Moves[C] + \Moves[D] \) is defined by:
\begin{inparaenum}[$(1)$]
\item if \( m \vdash_X m' \) for some \( X \in \{ A, B, C, D \} \), then \( m \vdash_{A, B, C, D} m' \),
\item if \( \star \vdash_D m \), then \( \star \vdash_{A, B, C, D} m \),
\item if \( \star \vdash_D m \) and \( \star \vdash_C m' \), then \( m \vdash_{A, B, C, D} m' \),
\item if \( \star \vdash_C m \) and \( \star \vdash_B m' \), then \( m \vdash_{A, B, C, D} m' \), and
\item if \( \star \vdash_B m \) and \( \star \vdash_A m' \), then \( m \vdash_{A, B, C, D} m' \).
\end{inparaenum}
A \emph{justified sequence} over \( (A, B, C, D) \) is a sequence of \( \Moves[A, B, C, D] \) equipped with pointers that respect \( \vdash_{A, B, C, D} \).  Given a justified sequence \( w \) over \( (A, B, C, D) \), the projections \( \proj{w}{A, B, C} \) onto interaction sequences and \( \proj{w}{A, B} \) onto plays are defined in the obvious way.  A justified sequence over \( (A, B, C, D) \) is an \emph{interaction sequence} if \( \proj{w}{A, B} \), \( \proj{w}{B, C} \) and \( \proj{w}{C, D} \) are plays and its last move is in \( \PMoves[A, D] = \OMoves[A] + \PMoves[D] \).
\end{definition}

\iflong 
\begin{definition}[Switching condition]
Let \( (A, B, C, D) \) be a quadruple of arenas and \( s \) be a sequence over \( \Moves[A, B, C, D] \).  It satisfies the \emph{switching condition} if it is accepted by the following automaton from the initial state \( \mathit{OOO} \) (all states are accepting).
\[\xymatrix@R-0.1cm@C2.4cm{
\mathit{OOO} \ar@<2pt>[r]^{\PMoves[A]} \ar@<1pt>[d]^{\OMoves[D]}
&
\mathit{POO} \ar@<2pt>[d]^{\PMoves[B]} \ar@<1pt>[l]^{\OMoves[A]}
\\
\mathit{OOP} \ar@<2pt>[u]^{\PMoves[D]} \ar@<1pt>[r]^{\OMoves[C]} &
\mathit{OPO} \ar@<2pt>[l]^{\PMoves[C]} \ar@<1pt>[u]^{\OMoves[B]} &
}
\]
The three components of states correspond to \( (A, B) \), \( (B, C) \) and \( (C, D) \) in this order.
\end{definition}

\begin{lemma}
Every interaction sequence over \( (A, B, C, D) \) satisfies the switching condition.
\end{lemma}
\begin{proof}
This is because the automaton checks if each component is O-P alternating.
\end{proof}
\fi 

A \emph{basic block} consists of the start move in \( \OMoves[A, D] = \PMoves[A] + \OMoves[D] \), the last move in \( \PMoves[A, D] \) and intermediate moves in \( \Moves[B] + \Moves[C] \).  An morphism \( f : w \to w' \) between interaction sequences over \( (A, B, C, D) \) is an injective map between move occurrences that preserve moves, the justification pointers and basic blocks.  We write \( \CIntr{A, B, C, D} \) for the category of generalised interaction sequences.

\begin{lemma}\quad
\begin{itemize}
\item Projections from \( \CIntr{A, B, C, D} \) (e.g.~\( \proj{}{A, B, C} \) and \( \proj{}{A, B} \)) are functors.
\item Composition of projections is a projection, e.g.
\[
  \CIntr{A, B, C, D} \stackrel{\proj{}{A, B, C}}{\longrightarrow} \CIntr{A, B, C} \stackrel{\proj{}{B, C}}{\longrightarrow} \CPlay{B, C}
  = \CIntr{A, B, C, D} \stackrel{\proj{}{B, C}}{\longrightarrow} \CPlay{B, C}.
\]
\item The projection \( \proj{}{A, D} : \CIntr{A, B, C, D} \to \CPlay{A, D} \) is a discrete fibration.
\end{itemize}
\end{lemma}
\iflong 
\begin{proof}
The first two claims are easy to see.  The third claim can be proved by the same technique to Lemma~\ref{lem:fibration}
\end{proof}
\fi 

\begin{lemma}\label{lem:interaction-pullback}
Let \( u \in \CIntr{A, B, D} \) and \( v \in \CIntr{B, C, D} \).  If \( (\proj{u}{B, D}) = (\proj{v}{B, D}) \), there exists a unique \( w \in \CIntr{A, B, C, D} \) such that \( u = \proj{w}{A, B, D} \) and \( v = \proj{w}{B, C, D} \).  A similar statement holds for every \( u \in \CIntr{A, C, D} \) and \( v \in \CIntr{A, B, C} \).
\end{lemma}
\iflong 
\begin{proof}
Let \( u = m_1 \dots m_M \in \CIntr{A, B, D} \) and \( v = n_1 \dots n_N \in \CIntr{B, C, D} \) and suppose that \( \pi^{B, D}(u) = \pi^{B, D}(v) \).  We construct \( w \in l_1 \dots l_L \in \CIntr{A, B, C, D} \).  By the switching condition, \( u \) and \( v \) must be accepted by the left and right automata, respectively,
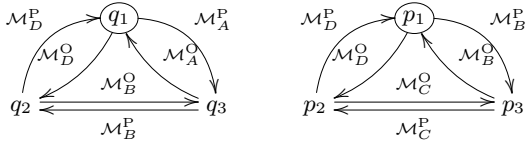
\[
\begin{array}{ccc}
\xymatrix@R-0.1cm@C.8cm{
& *+[o][F]{q_1} \ar@/^1pc/[dr]^{\PMoves[A]} \ar@/^/[dl]_{\OMoves[D]} & \\
q_2 \ar@/^1pc/[ur]^{\PMoves[D]} \ar@<1pt>[rr]^{\OMoves[B]} & &
q_3 \ar@<2pt>[ll]^{\PMoves[B]} \ar@/^/[ul]_{\OMoves[A]}
}
&
&
\xymatrix@R-0.1cm@C.8cm{
& *+[o][F]{p_1} \ar@/^1pc/[dr]^{\PMoves[B]} \ar@/^/[dl]_{\OMoves[D]} & \\
p_2 \ar@/^1pc/[ur]^{\PMoves[D]} \ar@<1pt>[rr]^{\OMoves[C]} & &
p_3 \ar@<2pt>[ll]^{\PMoves[C]} \ar@/^/[ul]_{\OMoves[B]}
}
\end{array}
\]
and \( w \) must be accepted by the automaton
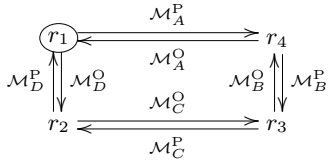
\[\xymatrix@R-0.1cm@C2.4cm{
*+[o][F]{r_1} \ar@<2pt>[r]^{\PMoves[A]} \ar@<1pt>[d]^{\OMoves[D]}
&
r_4 \ar@<2pt>[d]^{\PMoves[B]} \ar@<1pt>[l]^{\OMoves[A]}
\\
r_2 \ar@<2pt>[u]^{\PMoves[D]} \ar@<1pt>[r]^{\OMoves[C]} &
r_3 \ar@<2pt>[l]^{\PMoves[C]} \ar@<1pt>[u]^{\OMoves[B]} &
}
\]
We construct a sequence of moves \( w \) such that \( \proj{w}{A, B, D} = u \) and \( \proj{w}{B, C, D} = v \).  An \emph{intermediate state} is a tuple \( (i, j, k, p, q, r) \) such that \( i \le M \), \( j \le N \) such that \( \proj{m_i \dots m_M}{B, D} = \proj{n_j \dots n_N}{B, D} \), \( k \) is the current index of \( l \) and \( p \), \( q \) and \( r \) are states of the above automata from which \( m_i \dots m_M \), \( n_j \dots n_N \) and \( l_k \dots l_L \) are accepted, respectively.  \tk{TODO: Clarify what kind of induction is used.}
\begin{itemize}
\item \( (i, j, k, q_1, p_1, r_1) \):  Then \( m_i \in \OMoves[D] + \PMoves[A] \).  If \( m_i \in \OMoves[D] \), then let \( l_k = m_i = n_j \) and proceed to \( (i+1, j+1, k+1, q_2, p_2, r_2) \).  If \( m_i \in \PMoves[A] \), then let \( l_k = m_i \) and proceed to \( (i+1, j, k+1, q_3, p_1, r_4) \).
\item \( (i, j, l, q_2, p_2, r_2) \):  Then \( n_j \in \OMoves[C] + \PMoves[D] \).  If \( n_j \in \OMoves[C] \), then let \( l_k = n_j \) and proceed to \( (i, j+1, k+1, q_2, p_3, r_3) \).  If \( n_j \in \PMoves[D] \), then let \( l_k = n_j = m_i \) and proceed to \( (i+1, j+1, k+1, p_1, q_1, r_1) \).
\item \( (i, j, l, q_2, p_3, r_3) \):  Then \( n_j \in \OMoves[B] + \PMoves[C] \).  If \( n_j \in \OMoves[B] \), then let \( l_k = m_i = n_j \) and proceed to \( (i+1, j+1, k+1, q_3, p_1, r_4) \).  If \( n_j \in \PMoves[C] \), then let \( l_k = n_j \) and proceed to \( (i, j+1, k+1, q_2, p_2, r_2) \).
\item \( (i, j, l, q_3, p_1, r_4) \):  Then \( m_i \in \OMoves[A] + \PMoves[B] \).  If \( m_i \in \OMoves[A] \), then let \( l_k = m_i \) and proceed to \( (i+1, j, k+1, q_1, p_1, r_1) \).  If \( m_i \in \PMoves[B] \), then let \( l_k = m_i = n_j \) and proceed to \( (i+1, j+1, k+1, q_2, p_3, r_3) \).
\item Other cases are never reached.
\end{itemize}
The justification pointer for \( A \)-moves are determined by \( u \) and others by \( v \).
\end{proof}
\fi 

Given innocent strategies \( \sigma_1 \in \Sheaf(\CPlay{A, B}) \), \( \sigma_2 \in \Sheaf(\CPlay{B, C}) \) and \( \sigma_3 \in \Sheaf(\CPlay{C, D}) \), their \emph{simultaneous composition} \( F : \CPlay{A, D}^{\op} \to \Set \) is defined by: for objects, \( F(s) \) is
\[
  \coprod_{w \in \CIntr{A, B, C, D}:\,\proj{w}{A, D} = s} \sigma_1(\proj{w}{A, B}) \times \sigma_2(\proj{w}{B, C}) \times \sigma_3(\proj{w}{C, D})
\]
and, given \( f : s \to t \) in \( \CPlay{A, D} \), the function \( F(f) : F(t) \to F(s) \) maps \( (w, e_1, e_2, e_3) \in F(t) \) to
\[
  (f^*(w),\;\; e_1 \cdot (\proj{\bar{f}_w}{A, B}),\;\; e_2 \cdot (\proj{\bar{f}_w}{B, C}),\;\; e_3 \cdot (\proj{\bar{f}_w}{C, D})).
\]

\begin{lemma}
The simultaneous composition is naturally isomorphic to sequential compositions \( \sigma_1; (\sigma_2; \sigma_3) \) and \( (\sigma_1; \sigma_2); \sigma_3 \).
\end{lemma}
\begin{proof}
Given \( s \in \CPlay{A, D} \), consider a function \( \psi_s \) that maps an element \( (w, e_1, e_2, e_3) \) of
\[
  \coprod_{w:\,\proj{w}{A, D} = s} \sigma_1(\proj{w}{A, B}) \times \sigma_2(\proj{w}{B, C}) \times \sigma_3(\proj{w}{C, D})
\]
to \( ((\proj{w}{A, B, D}), e_1, ((\proj{w}{B, C, D}), e_2, e_3)) \) of
\[
  \coprod_{u:\,\proj{u}{A, D} = s} \sigma_1(\proj{u}{A, B}) \times \coprod_{u:\,\proj{v}{B, D} = \proj{u}{B, D}} \sigma_2(\proj{v}{B, C}) \times \sigma_3(\proj{v}{C, D}).
\]
This is a bijection thanks to Lemma~\ref{lem:interaction-pullback}.  It is easy to show the naturality of \( \psi \).
\iflong 

Let us write \( F \) for the simultaneous composition and \( G = (\sigma_1; (\sigma_2; \sigma_3)) \).  Assume \( f : s \to t \) in \( \CPlay{A, D} \).  Then \( F(f); \psi_t \) maps \( (w, e_1, e_2, e_3) \) to
\begin{align*}
  &
  (w, e_1, e_2, e_3)
  \\ \stackrel{F(f)}{\longmapsto}
  &
  (f^*(w),\;\; e_1 \cdot (\proj{\bar{f}_w}{A, B}),\;\; e_2 \cdot (\proj{\bar{f}_w}{B, C}),\;\; e_3 \cdot (\proj{\bar{f}_w}{C, D}))
  \\ \stackrel{\psi_t}{\longmapsto}
  &
  \begin{array}{l}
    (\proj{f^*(w)}{A, B, D},\;\; e_1 \cdot (\proj{\bar{f}_w}{A, B}), \\
    \qquad (\proj{f^*(w)}{B, C, D}),\;\;  e_2 \cdot (\proj{\bar{f}_w}{B, C}),\;\; e_3 \cdot (\proj{\bar{f}_w}{C, D}))
  \end{array}
\end{align*}
and \( \psi_s; G(f) \) maps \( (w, e_1, e_2, e_3) \) to
\begin{align*}
  &
  (w, e_1, e_2, e_3)
  \\ \stackrel{\psi_s}{\longmapsto}
  &
  ((\proj{w}{A, B, D}),\; e_1,\;((\proj{w}{B, C, D}),\; e_2,\; e_3))
  \\[5pt] \stackrel{G(f)}{\longmapsto}
  &
  \begin{array}{l}
    (f^*(\proj{w}{A, B, D}),\;\; e_1 \cdot (\proj{\bar{f}_{\proj{w}{A, B, D}}}{A, B}), \\
    \qquad
    ((\proj{\bar{f}_{\proj{w}{A, B, D}}}{B, D})^*(\proj{w}{B, C, D}), \\
    \qquad\;\;
    e_2 \cdot (\proj{\overline{(\proj{\bar{f}_{\proj{w}{A, B, D}}}{B, D})}_{(\proj{w}{B, C, D})}}{B, C}), \\
    \qquad\;\;
    e_3 \cdot (\proj{\overline{(\proj{\bar{f}_{\proj{w}{A, B ,D}}}{B, D})}_{(\proj{w}{B, C, D})}}{C, D}))).
  \end{array}
\end{align*}
By Lemma~\ref{lem:quad-interaction-projection}, we have \( \proj{f^*(w)}{A, B, D} = f^*(\proj{w}{A, B, D}) \), so the first components coincide.  As for the second components, again by Lemma~\ref{lem:quad-interaction-projection}, we have
\[
  \proj{(\bar{f}_{\proj{w}{A, B, D}})}{A, B} = \proj{\proj{(\bar{f}_w)}{A, B, D}}{A, B} = \proj{\bar{f}_w}{A, B}.
\]
For the third components, recall that
\[
  \proj{(\bar{f}_{\proj{w}{A, B, D}})}{B, D} = \proj{(\bar{f}_w)}{B, D}.
\]
and \( \proj{(\bar{f}_w)}{B, D} : (\proj{f^*(w)}{B, D}) \to (\proj{w}{B, D}) \).  Since
\[
  \proj{(\bar{f}_w)}{B, C, D} : (\proj{f^*(w)}{B, C, D}) \to (\proj{w}{B, C, D})
\]
is projected onto \( \proj{(\bar{f}_w)}{B, D} \), we have
\[
  (\proj{(\bar{f}_w)}{B, D})^*(\proj{w}{B, C, D}) = \proj{f^*(w)}{B, C, D}.
\]
For the fourth components, by using Lemma~\ref{lem:quad-interaction-projection}, we have
\begin{align*}
  \overline{(\proj{\bar{f}_{\proj{w}{A, B, D}}}{B, D})}_{\proj{w}{B, C, D}}
  &= \overline{(\proj{\bar{f}_{w}}{B, D})}_{\proj{w}{B, C, D}}
  \\
  &= \overline{(\proj{\proj{\bar{f}_{w}}{B, C, D}}{B, D})}_{\proj{w}{B, C, D}}
  \\
  &= \proj{(\bar{f}_{w})}{B, C, D}
\end{align*}
(in general, for \( h : u \to v \) in \( \CIntr{B, C, D} \), we have \( \overline{(\proj{h}{B, D})}_v = h \))
and
\[
  \proj{\proj{(\bar{f}_{w})}{B, C, D}}{B, C} = \proj{\bar{f}_w}{B, C}
\]
as desired.  The fifth component is the same.
\fi 
\end{proof}

\iflong 
\begin{lemma}\label{lem:quad-interaction-projection}
Let \( w \in \CIntr{A, B, C, D} \) and \( f : s \to (\proj{w}{A, D}) \) in \( \CPlay{A, D} \).  Then
\[
  \proj{f^*(w)}{A, B, D} = f^*(\proj{w}{A, B, D})
\]
and
\[
  \proj{\bar{f}_{w}}{A, B, D} = \bar{f}_{\proj{w}{A, B, D}}.
\]
\end{lemma}
\begin{proof}
By definition, \( \bar{f}_w : f^*(w) \to w \) in \( \CIntr{A, B, C, D} \).  Thus
\[
  \proj{\bar{f}_w}{A, B, D} : (\proj{f^*(w)}{A, B, D}) \to (\proj{w}{A, B, D}).
\]
Both claims follow from \( \proj{\proj{\bar{f}_w}{A, B, D}}{A, D} = \proj{\bar{f}_w}{A, D} = f \).
\end{proof}
\fi 

\begin{corollary}
Composition is associative up to isomorphism.
\end{corollary}

\subsection{CCC of arenas and strategies}

\begin{definition}
The \emph{category of arenas and strategies} \( \CGame \) has arenas as objects and a sheaf \( \sigma \in \Sheaf(\CPlay{A, B}) \) as a morphism from \( A \) to \( B \).
\tkchanged{We regard that isomorphic sheaves define the same morphism.}  The composition is defined in Section~\ref{sec:composition}.
\end{definition}

As usual, the identity morphisms are copycat strategies.
\begin{definition}
Let \( A \) be an arena.  Let us write a move in \( \Moves[A, A] = \Moves[A] + \Moves[A] \) as \( l(m) \) and \( r(m) \) for \( m \in \Moves[A] \), in order to distinguish the component.  The relation \( \sim \) is given by \( l(m) \sim r(m) \) and \( r(m) \sim l(m) \) (i.e.~\( \sim \) relates the same move in the different component).  A play \( s = m_1 m_2 \dots m_n \in \CPlay{A, A} \) is \emph{copycat} if, for every even number \( k \le n \),
\begin{inparaenum}[$(1)$]
\item \( m_{k-1} \sim m_{k} \),
\item \( \star \curvearrowleft m_{k-1} \) implies \( m_{k-1} \curvearrowleft m_k \), and
\item \( m_j \curvearrowleft m_{k-1} \) implies \( m_{j-1} \curvearrowleft m_k \).
\end{inparaenum}
The \emph{copycat strategy} \( \ident_A \in \Sheaf(\CPlay{A, A}) \) is defined by: \( \ident_A(s) = \{ \ast \} \) if \( s \) is copycat and \( \ident_A(s) = \emptyset \) otherwise.
\end{definition}

\begin{proposition}
\( (\ident_A; \sigma) \cong \sigma \cong (\sigma; \ident_B) \) for \( \sigma \in \Sheaf(\CPlay{A, B}) \).
\end{proposition}

\tkchanged{In the rest of this subsection, we show that \( \CGame \) is a CCC.  It is an adaptation of the standard arguments for HO/N game models.}

\paragraph{Products and terminal object}
Given arenas $A$ and $B$, the arena \( A \times B \) is defined by: \( \Moves[A \times B] := \Moves[A] + \Moves[B] \), \( \lambda_{A \times B} := [\lambda_A, \lambda_B] \) and \( (\vdash_{A \times B}) := (\vdash_A) \cup (\vdash_B) \).  We say a play \( s \in \CPlay{A \times B, A} \) is \emph{copycat} if \( s \) does not contain \( B \)-moves and it is copycat as a play of \( \CPlay{A, A} \).  The projection \( \pi_1 \in \Sheaf(\CPlay{A \times B, A}) \) is defined by: \( \pi_1(s) = \{ \ast \} \) if \( s \) is copycat and \( \pi_1(s) = \emptyset \) otherwise.  The projection \( \pi_2 \in \Sheaf(\CPlay{A \times B, B}) \) is defined similarly.

\tkchanged{For a play \( s \in \CPlay{A, B \times C} \), we write \( \proj{s}{A, B} \) for the restriction of \( s \) to \( \{ i \mid m_j \curvearrowleft^* m_i \textrm{ for some } m_j \in \Moves[B] \} \), where \( \curvearrowleft^* \) is the reflexive and transitive closure of \( \curvearrowleft \).  The restriction is a functor.}

The terminal object is the empty arena having no moves.



\paragraph{Exponentials}
Let \( A \) and \( B \) be arenas.  The \emph{exponential arena} \( A \Rightarrow B \) is defined by:
\begin{inparaenum}[$(1)$]
\item \( \Moves[A \Rightarrow B] := \{ m \in \Moves[B] \mid \star \vdash_B m \} \times \Moves[A] + \Moves[B] \),
\item \( \lambda_{A \Rightarrow B}(m) := \neg \lambda_{A}(m_A) \) (if \( m = (m_B, m_A) \)) and \( \lambda_{A \Rightarrow B}(m) := \lambda_{B}(m) \) (if \( m \in \Moves[B] \)), where \( \neg \OO = \PP \) and \( \neg \PP = \OO \).
\end{inparaenum}
The enabling relation is defined by
\begin{inparaenum}[$(a)$]
\item if \( \star \vdash_A m_A \), then \( m_B \vdash_{A \Rightarrow B} (m_B, m_A) \),
\item if \( m_A \vdash_A m_A' \), then \( (m_B, m_A) \vdash_{A \Rightarrow B} (m_B, m_A') \), and
\item if \( m_B \vdash_B m_B' \), then \( m_B \vdash_{A \Rightarrow B} m_B' \).
\end{inparaenum}

Given a play \( s \in \CPlay{A, B \Rightarrow C} \), let us write \( \theta(s) \) for the justified sequence in which \( (m_C, m_B) \in \Moves[B \Rightarrow C] \subseteq \Moves[A, B \Rightarrow C] \) is replaced with \( m_B \in \Moves[A \times B, C] \).  Then \( \theta(s) \) is a play over \( (A \times B, C) \).  Conversely, given a play \( s \in \CPlay{A \times B, C} \), let us write \( \theta^{-1}(s) \) for the justified sequence in which every B-move \( m_B \in \Moves[B] \subseteq \Moves[A \times B, C] \) is replaced with \( (m_C, m_B) \in \Moves[B \Rightarrow C] \subseteq \Moves[A, B \Rightarrow C] \) where \( m_C \) is the initial \( C \)-move s.t.~\( m_C \curvearrowleft^+ m_B \). 
Since \( \theta \) and \( \theta^{-1} \) do not change the order of move occurrences nor justification pointers, they are functors.  Furthermore \( \theta^{-1} \) is the inverse of \( \theta \).  So \( \CPlay{A, B \Rightarrow C} \) is isomorphic to \( \CPlay{A \times B, C} \).  Since \( \theta \) maps views to views, we have an isomorphism between \( \CView{A, B \Rightarrow C} \) and \( \CView{A \times B, C} \) as well.

The isomorphism \( \theta : \CPlay{A, B \Rightarrow C} \to \CPlay{A \times B, C} \) gives an isomorphism \( \Lambda : \Sheaf(\CPlay{A \times B, C}) \to \Sheaf(\CPlay{A, B \Rightarrow C}) : \sigma \mapsto \sigma \circ \theta \).  This is a natural bijection on hom-sets \( \Lambda: \CGame(A \times B, C) \cong \CGame(A, B \Rightarrow C) \).


In summary, we have the following result.
\begin{lemma}
\( \CGame \) is a cartesian closed category.
\end{lemma}

\subsection{Key lemma for full completeness}\label{sec:key-lemma}
Basically the full completeness is achieved by establishing the correspondence between the paths of terms in normal form and P-views.  \tkchanged{This subsection describes the key lemma for full completeness, adapting the standard technique for HO/N game models.}

An arena \( A \) is prime if it has a unique initial move.  Then \( A = B \Rightarrow \{ m \} \) for some arena \( B \) and the initial \( A \)-move \( m \).

Let \( A = A_1 \times \dots \times A_n \) be an arena, where \( A_i \) is prime for each \( i \), and \( i \in [n] \).  Writing \( m_2 \) for the unique initial \( A_i \)-move, \( (m_1 m_2) \in \CView{A, \{ m_1 \}} \).  We define \( (m_1 m_2) / \CView{A, \{ m_1 \}} \) as the full subcategory consisting of P-views \( p > (m_1 m_2) \).  (Since \( \CView{A, \{ m_1 \}} \) is a poset, this coincides with the standard definition of the under category.)  Suppose \( A_i = B \Rightarrow \{ m_2 \} \).  There is an isomorphism
\[
  \chi_{(m_1m_2)} : (m_1 m_2) / \CView{A, \{m_1\}} \stackrel{\cong}{\to} \CView{A, B},
\]
given by \( m_1 m_2 m_3 \dots m_l \mapsto m_3 \dots m_l \).  Here we need to modify the justification pointer as follows:
\begin{itemize}
\item If \( m_2 \curvearrowleft m_k \) in LHS (then \( k = 3 \)), then \( \star \curvearrowleft m_k \) in RHS.
\item If \( m_1 \curvearrowleft m_k \) in LHS, then \( m_3 \curvearrowleft m_k \) in RHS.
\item If \( m_j \curvearrowleft m_k \) in LHS (\( j \neq 1, 2 \)), then \( m_j \curvearrowleft m_k \) in RHS.
\end{itemize}
This isomorphism is the key to prove full completeness.

Let \( \tau \in \Sheaf(\CView{A, B}) \).  Suppose that \( A = A_1 \times \dots A_n \), where \( A_i \) is prime for each \( i \).  Let \( i \in [n] \) and \( A_i = B \Rightarrow \{ m_2 \} \).  We define the operation \( (m_1m_2) \rhd \tau \) that ``inserts'' \( m_1 m_2 \) before the P-views in \( \tau \), defined by:
\[
\begin{array}{llr}
  ((m_1 m_2) \rhd \tau)(m_1 m_2) &:= \{ \ast \} \\
  ((m_1 m_2) \rhd \tau)(m_1 m_2 p) &:= \tau(p) \\ 
  ((m_1 m_2) \rhd \tau)(p) &:= \emptyset & \textrm{ (otherwise)}.
\end{array}
\]
To be precise, the second equation should be written as \( ((m_1m_2) \rhd \tau)(m_1m_2p) := \tau(\chi_{(m_1m_2)}(m_1m_2p)) \).
Then \( ((m_1 m_2) \rhd \tau) \in \Sheaf(\CView{A, \{ m_1 \}}) \).

\begin{lemma}\label{lem:key-lemma}
Let \( \tau \in \Sheaf(\CView{A, B}) \) and suppose that \( A = A_1 \times \dots \times A_n \), \( A_i \) is prime for all \( i \), \( k \in [n] \) and \( A_k = B \Rightarrow \{ m_2 \} \).  Then
\[
  \iota_*((m_1 m_2) \rhd \tau) \cong \langle \pi_i, \iota_*(\tau) \rangle; \mathbf{ev}
\]
where \( \pi_i \in \Sheaf(\CPlay{A, A_i}) \) is the projection of the product and \( \mathbf{ev} = \Lambda(\ident_{A_i}) \in \Sheaf(\CPlay{(B \Rightarrow \{ m_2 \}) \times B,\;\{ m_2 \}}) \) is the evaluation map.
\end{lemma}



\section{Sheaves model for deterministic \( \lambda_{\to} \)}
This section develops the sheaves model for simply-typed \( \lambda \)-calculus, the simplest functional programming language.

\subsection{The target language}
The standard simply-typed call-by-name \( \lambda \)-calculus extended to have divergence \( \bot \).  The syntax of terms is given by:
\[
  M ::= x \mid \lam x. M \mid M\,M \mid \bot.
\]
We consider simply-typed terms possibly having free variables.  Types are type environments are given by the grammar:
\[
  \kappa ::= \T \mid \kappa \to \kappa
\qquad
  \Gamma ::= \cdot \mid \Gamma, x : \kappa.
\]
The typing rules are standard, expect that \( \bot \) is considered as a constant of the ground type \( \T \).

We study the equational theory of terms, precisely \( \beta\eta \)-theory.  The relation \( = \) is the least equivalence relation that satisfies
\begin{align*}
  (\lam x. M)\,N &= M[N/x] \\
  \lam x. M\,x &= M & \textrm{ (\( x \) fresh)}
\end{align*}
and the congruence rules: if \( M = M' \), then \( M\,N = M'\,N \) and \( N\,M = N\,M' \).  The normal form is defined by:
\[
  Q ::= \lam x_1 \dots x_k. y\,Q_1\,\dots\,Q_n \mid \lam x_1 \dots x_k. \bot
\]
where \( y\,Q_1\,\dots\,Q_n \) is fully applied, i.e.~\( y\,Q_1\,\dots\,Q_n : \T \).  Every term has a unique normal form.

\subsection{Deterministic strategies}

\begin{definition}
An \emph{odd-length play} is an odd-length alternating P-visible justified sequence.  (It is not a play because a play is of even-length.)  For an odd-length play \( s \) over \( (A, B) \), the \emph{immediate extension} \( \ImmExt{s} \) is a set of plays \( \{ sm \mid sm \in \CPlay{A, C} \} \).
\end{definition}

An odd-length play \( s \) ends with an O-move and the immediate extension \( \ImmExt{s} \) is the set of all possible Proponent's responses.

\begin{definition}
An innocent strategy \( \sigma \in \Sheaf(\CPlay{A, B}) \) is \emph{deterministic} if, for every odd-length play \( s \), \( \coprod_{t \in \ImmExt{s}} \sigma(t) \) is empty or singleton.
It is \emph{finite} if \( \{ p \in \CView{A, B} \mid \sigma(p) \neq \emptyset \} \) is a finite set.
\end{definition}

\begin{remark}
If \( \sigma \) is deterministic, then \( \sigma(s) \) is empty or singleton for every \( s \in \CPlay{A, B} \).  So it is completely determined by a set \( \{ s \in \CPlay{A, B} \mid \sigma(s) \neq \emptyset \} \).  Through this translation, the sheaf-based definition of innocent strategies coincides with the standard one.
\end{remark}

\begin{definition}
A \emph{category of deterministic strategies} \( \CDetGame \) is a subcategory consisting of deterministic strategies.
\end{definition}

\( \CDetGame \) is well-defined since the identity \( \ident_A \) deterministic and the composition preserves determinacy.
\begin{lemma}
Composition preserves determinacy.
\end{lemma}
\begin{proof}
Let \( \sigma_1 \in \Sheaf(\CPlay{A, B}) \) and \( \sigma_2 \in \Sheaf(\CPlay{B, C}) \) be deterministic strategies.  Then for every odd-length play \( s \) of \( (A, C) \), there exists at most one \( u \) such that \( \proj{u}{A, C} = s m \), \( \sigma_1(\proj{u}{A, B}) \neq \emptyset \) and \( \sigma_2(\proj{u}{B, C}) \neq \emptyset \) (see \emph{uncovering construction} in \cite{HylandO00}).  Thus \( \coprod_{u:\,\proj{u}{A, C} \in \ImmExt{s}} \sigma_1(\proj{u}{A, B}) \times \sigma_2(\proj{u}{B, C}) \) is empty or singleton.
\end{proof}

Since projections \( A \times B \to A \) and \( A \times B \to B \) are deterministic and the isomorphism \( \Sheaf(\CPlay{A \times B, C}) \cong \Sheaf(\CPlay{A, B \Rightarrow C}) \) preserves determinacy, \( \CDetGame \) is a CCC.

\subsection{Interpretation}
Simple types are interpreted as objects by
\[
  \sem{\T} := \{ m_{\T} \}
\qquad
  \sem{\kappa \to \kappa'} := \sem{\kappa} \Rightarrow \sem{\kappa'}
\]
as well as type environemnts
\[
  \sem{x_1 : \kappa_1, \dots, x_n : \kappa_n} := \sem{\kappa_1} \times \dots \times \sem{\kappa_n}.
\]
The interpretation of terms is fairly standard:
\begin{align*}
  \sem{x_1 : \kappa_1, \dots, x_n : \kappa_n \vdash x_i : \kappa_i} &:= \pi_i
  \\
  \sem{\Gamma \vdash \lam x. M : \kappa \to \kappa'} &:= \Lambda(\sem{\Gamma, x : \kappa \vdash M : \kappa'})
  \\
  \sem{\Gamma \vdash M\,N : \kappa} &:= \langle \sem{M}, \sem{N} \rangle; \mathbf{ev}
  \\
  \sem{\Gamma \vdash \bot : \T} &:= \iota_* \tau_{\emptyset},
\end{align*}
where \( \tau_{\emptyset} \in \Sheaf(\CView{\sem{\Gamma}, \sem{\T}}) \) is the constant functor mapping to \( \emptyset \).

\begin{theorem}[Soundness]
\( M = N \) iff \( \sem{M} \cong \sem{N} \).
\end{theorem}
\begin{proof}
This is a special case of Theorem~\ref{thm:nondet-sound} below.
\end{proof}


\begin{theorem}[Full completeness]
Let \( \Gamma \) be a type environment, \( \kappa \) be a simple type and \( \sigma \in \Sheaf(\CPlay{\sem{\Gamma}, \sem{\kappa}}) \).  If \( \sigma \) is finite and deterministic, there exists a term \( \Gamma \vdash M : \kappa \) such that \( \sigma \cong \sem{M} \).
\end{theorem}
\begin{proof}
The set \( \makeset{p \in \CView{\sem{\Gamma}, \sem{\kappa}} \mid \sigma(p) \neq \emptyset} \), which is finite and prefix-closed, gives a finite \emph{view function} in the sense of \cite{HylandO00}.  A term \( M \) that denotes \( \sigma \) can then be constructed by induction on the size of the view function, following the proof of Prop.~7.4 in \emph{op.~cit.}.  \tkchanged{One can directly construct a term \( M \) using Lemma~\ref{lem:key-lemma}.}
\end{proof}

\section{Sheaves model for nondeterministic \( \lambda_{\to} \)}
This section studies an extension of \( \lambda_\to \) having the non-deterministic branch and interprets the calculus using \( \CGame \).  We shall prove the soundness of interpretation and the full completeness.

\subsection{The target language}
Consider the simply-typed lambda calculus with \( \bot \) extended to have the non-deterministic branch: \( M_1 + M_2 \).  The additional axioms are:
\begin{align*}
  (M_1 + M_2)\,N &= (M_1\,N) + (M_2\,N) \\
  \lam x. (M_1 + M_2) &= (\lam x. M_1) + (\lam x. M_2) \\
  M + (\lambda x_1 \dots x_n. \bot) &= M
\end{align*}
and the associativity and commutativity of \( + \).  These equations are sound with respect to the observational equivalence in the call-by-name evaluation strategy, where the observable is may-convergence.  (They are not sound for must-convergence because of the right equation.)

We define \emph{normal forms} where \( n, k \geq 0 \):
\[
  R := Q_1 + \dots + Q_n
\qquad
  Q := \lambda x_1 \ldots x_n . y \, R_1 \ldots R_k,
\]
where \( y\,R_1 \dots R_k \) is fully applied.  Every term has a unique normal form (modulo the commutation of non-deterministic branches), or is equivalent to \( \lam x_1 \dots x_n. \bot \).  Note that \( M + M \neq M \) in general.

\subsection{Interpretation and soundness}
The term \( \Gamma \vdash M + N : \kappa \) is interpreted as the coproduct \( \sem{M} + \sem{N} \) in \( \Sheaf(\CPlay{\sem{\Gamma}, \sem{\kappa}}) \).  A simple way to describe the coproduct is to use sheaves over views: since the sheaves over views are just presheaves, the coproduct can be computed pointwise.  So, given \( \tau_1, \tau_2 \in \Sheaf(\CView{A, B}) \), we have \( (\tau_1 + \tau_2)(p) = \tau_1(p) + \tau_2(p) \).  For sheaves \( \sigma_1, \sigma_2 \in \Sheaf(\CPlay{A, B}) \) over plays, we define \( \sigma_1 + \sigma_2 := \iota_* ((\iota^* \sigma_1) + (\iota^* \sigma_2)) \) using the Comparison Lemma (Lemma~\ref{lem:comparison}).

\tkchanged{Coproducts on the function position commutes with application.}
\begin{lemma}\label{lem:nondet-func-coprod}
\( (\langle \sigma_0, \sigma_1 + \sigma_2 \rangle; \mathbf{ev}) \cong (\langle \sigma_0, \sigma_1 \rangle; \mathbf{ev}) + (\langle \sigma_0, \sigma_2 \rangle; \mathbf{ev}) \).
\end{lemma}
\begin{proof}(Sketch)
By Lemma~\ref{lem:comparison}, it suffices to consider their restrictions on P-views.  Let \( \sigma_0 \in \CPlay{A, B} \) and \( \sigma_1, \sigma_2 \in \CPlay{A, B \Rightarrow C} \).  Let us write \( D := B \times (B \Rightarrow C) \) for simplicity.  Then the right-hand-side on P-view \( p \in \CView{A, C} \) is given by
\begin{align*}
  &
  ((\langle \sigma_0, \sigma_1 \rangle; \mathbf{ev}) + (\langle \sigma_0, \sigma_2 \rangle; \mathbf{ev}))(p) \\
  \cong&
  \coprod_{u: \proj{u}{A, C} = p} \sigma_0(\proj{u}{A, B}) \times \sigma_1(\proj{u}{A, B \Rightarrow C}) \times \mathbf{ev}({\proj{u}{D, C}}) \\
  &
  + \coprod_{u: \proj{u}{A, C} = p} \sigma_0(\proj{u}{A, B}) \times \sigma_2(\proj{u}{A, B \Rightarrow C}) \times \mathbf{ev}({\proj{u}{D, C}}) \\
  \cong & \hspace{-8pt}
  \coprod_{u: \proj{u}{A, C} = p} \hspace{-12pt} \sigma_0(\proj{u}{A, B}) \!\times\! (\sigma_1(\proj{u}{A, B \Rightarrow C}) \!+\! \sigma_2(\proj{u}{A, B \Rightarrow C})) \!\times\! \mathbf{ev}(\proj{u}{D, C}).
\end{align*}
A play is \emph{well-opened} if it has exactly one move pointing to \( \star \).  If \( \proj{u}{A, C} \) has a unique initial \( C \)-move and \( \mathbf{ev}(\proj{u}{B \times (B \Rightarrow C), C}) \neq \emptyset \), then \( \proj{u}{B \Rightarrow C} \) has a unique initial \( (B \Rightarrow C) \)-move and hence \( \proj{u}{A, B \Rightarrow C} \) is well-opened.  So we can assume without loss of generality that \( \proj{u}{A, B \Rightarrow C} \) ranges over well-opened plays.  We claim that for a well-opened play \( \proj{u}{A, B \Rightarrow C} \), we have a bijection on sets
\[
  \sigma_1(\proj{u}{A, B \Rightarrow C}) + \sigma_2(\proj{u}{A, B \Rightarrow C}) \cong (\sigma_1 + \sigma_2)(\proj{u}{A, B \Rightarrow C}).
\]
The required natural isomorphism is the consequence of the claim.  Assume \( \proj{u}{A, B \Rightarrow C} = s = m_1 \dots m_n \) and let \( p_k := \PView{m_1 \dots m_k} \) (for \( k \in \{ 2, 4, \dots, n \} \)) and \( \{ f_{k} : p_k \to s \}_{k \in \{ 2, 4, \dots, n \}} \) be a covering family.  Then \( (\iota_* ((\iota^* \sigma_1) + (\iota^* \sigma_2)))(s) \) is the set of sequences of the form \( e_2 \dots e_n \), where \( e_k \in \sigma_1(p_k) + \sigma_2(p_k) \).  Since \( s \) is well-opened, \( f_2 : p_2 \to s \) is factor through \( f_k : p_k \to s \) for every \( k \).  This means that \( e_k \)'s come from the same component as \( e_2 \).  So \( e_k \in \sigma_1(p_k) \) for all \( k \) or \( e_k \in \sigma_2(p_k) \) for all \( k \).  Hence \( (\iota_* ((\iota^* \sigma_1) + (\iota^* \sigma_2)))(\proj{u}{A, B \Rightarrow C}) \) has a bijection to \( (\iota_* \iota^* \sigma_1)(\proj{u}{A, B \Rightarrow C}) + (\iota_* \iota^* \sigma_2)(\proj{u}{A, B \Rightarrow C}) \) as desired.
\end{proof}

Let us write \( \sem{M}_V \) for its the restriction on views, i.e.~\( \iota^* \sem{M} \).  For a term in normal form, its view restriction can be computed by the induction on the structure.  By definition,
\[
  \sem{\Gamma \vdash \lambda x. Q : \kappa \to \kappa'}_V \cong \sem{M}_V \circ \theta,
\]
where \( \theta : \CView{\sem{\Gamma}, \sem{\kappa \to \kappa'}} \to \CView{\sem{\Gamma \times \kappa}, \sem{\kappa'}} \) is the isomorphism, and \( \sem{\Gamma \vdash Q_1 + \dots + Q_n : \kappa}_V \cong \sem{Q_1}_V + \dots + \sem{Q_n}_V \).  The next lemma gives the interpretation of head variable, which is a consequence of Lemma~\ref{lem:key-lemma}.

\begin{lemma}\label{lem:nondet-interpret-head-variable}
Assume a term \( \Gamma \vdash x_i\,R_1\,\dots\,R_n : \T \) where \( x_i : \kappa_i \in \Gamma \).  Let \( m_1 \) be the unique initial move of \( \sem{\T} \) and \( m_2 \) be the unique initial move of \( \sem{\kappa_i} \).  Then
\[
  \sem{\Gamma \vdash x_i\,R_1\,\dots\,R_n : \T}_V \cong (m_1 m_2) \rhd \langle \sem{R_1}_V, \dots, \sem{R_n}_V \rangle.
\]
\end{lemma}

Let \( B \) be a prime arena and \( m_1 \) be the unique initial move.  A sheaf \( \sigma \in \Sheaf(\CPlay{A, B}) \) is \emph{deterministic on initial response} if \( \coprod_{s \in \ImmExt{m_1}} \sigma(s) \) is singleton.
\begin{lemma}\label{lem:nondet-head-initial-response}
\( \sem{Q} \) is deterministic on initial response.
\end{lemma}
\iflong 
\begin{proof}
By induction on the structure of \( Q \).  If \( Q = x_i\,R_1\,\dots\,R_n \), this follows from Lemma~\ref{lem:nondet-interpret-head-variable}.  If \( Q = \lam x. Q' \), then \( \sem{Q'} \) is deterministic on initial response and \( \Lambda : \Sheaf(\CPlay{A \times B, C}) \to \Sheaf(\CPlay{A, B \Rightarrow C}) \) preserves this property.  Hence \( \sem{Q} = \Lambda(\sem{Q'}) \) is deterministic on the initial response.
\end{proof}
\fi 

\begin{theorem}[Soundness]\label{thm:nondet-sound}
\( M = N \) iff \( \sem{M} \cong \sem{N} \).
\end{theorem}
\begin{proof}
To prove the left-to-right direction, it suffices to show the all equations are valid.  The equation \( \sem{(M_1 + M_2)\,N} \cong \sem{(M_1\,N) + (M_2\,N)} \) follows from Lemma~\ref{lem:nondet-func-coprod}.  Because \( + \) is the coproduct, it is commutative and associative.  Because \( \iota^* \sem{\bot} \) is the constant functor to \( \emptyset \), we have \( \sigma + \sem{\bot} \cong \sigma \) for every \( \sigma \).  \tk{How about \( \beta \)?}

To prove the converse, assume that \( \sem{M} \cong \sem{N} \) for normal terms \( \Gamma \vdash M : \kappa \) and \( \Gamma \vdash N : \kappa \).  Let \( m_1 \) be the unique initial move of \( \sem{\kappa} \).  Then, since \( \sem{M} \cong \sem{N} \), we have a bijection between \( \coprod_{s \in \ImmExt{m_1}} \sem{M}(s) \) and \( \coprod_{s \in \ImmExt{m_1}} \sem{N}(s) \).  Let \( n \) be the number of elements of those sets.  Then \( M \equiv Q_1 + \dots + Q_n \) since \( \sem{Q_i} \) is deterministic on initial response for every \( i \in [n] \) by Lemma~\ref{lem:nondet-head-initial-response}.  Similarly \( N \equiv Q'_1 + \dots + Q'_n \).  Since \( \sem{M} \cong \sem{N} \), there is a bijection \( \phi : [n] \to [n] \) such that \( \sem{Q_i} \cong \sem{Q_{\phi(i)}} \).  By the induction hypothesis, \( Q_i = Q_{\phi(i)} \).  So \( M = N \).

Suppose that
\begin{align*}
  M &\equiv \lambda x_1 \dots x_k. y\,R_1\,\dots\,R_n
  \\
  N &\equiv \lambda x_1 \dots x_k. y'\,R'_1\,\dots\,R'_{n'}
\end{align*}
We can assume without loss of generality that \( k = 0 \).  Then by Lemma~\ref{lem:nondet-interpret-head-variable}, we have
\[
  \sem{M}_{V} \cong (m_1 m_2) \rhd \langle \sem{R_1}_V, \dots, \sem{R_n}_V \rangle
\]
and
\[
  \sem{N}_{V} \cong (m_1 m'_2) \rhd \langle \sem{R'_1}_V, \dots, \sem{R'_n}_V \rangle
\]
where \( m_2 \) is the initial move for \( y : \kappa' \in \Gamma \) and \( m_2' \) is the initial move of \( y' : \kappa'' \in \Gamma \).  Since \( \sem{M}_V \cong \sem{N}_V \), we have \( m_2 = m'_2 \), which implies \( y = y' \) and \( n = n' \).  Furthermore \( \sem{M}_V \cong \sem{N}_V \) implies \( \sem{R_i}_V \cong \sem{R'_i}_V \) for all \( i \in [n] \) and thus \( \sem{R_i} \cong \sem{R'_i} \).  By the induction hypothesis, \( R_i = R'_i \) and hence \( M = N \).
\end{proof}

\subsection{Full completeness}

A sheaf \( \sigma \in \Sheaf(\CPlay{A, B}) \) is \emph{finite} if \( \coprod_{p \in \CView{A, B}} \sigma(\iota(p)) \) is finite.

\begin{lemma}\label{lem:nondet-unique-initial-decompose}
Every finite sheaf \( \sigma \in \Sheaf(\CPlay{A, \{ m_1 \}}) \) can be decomposed as \( \sigma \cong \sigma_1 + \dots + \sigma_n \), where \( \sigma_i \) is deterministic on initial response for all \( i \).
\end{lemma}
\begin{proof}
Let \( \tau = \iota^* \sigma \) be the restriction of \( \sigma \) to views. 
Consider the finite set \( \coprod_{p \in \ImmExt{m_1}} \tau(p) \), which we write as \( \{ (p_1, a_1), \dots, (p_n, a_n) \} \) (\( a_i \in \sigma(p_i) \) for each \( i \in [n] \)).  We define \( \tau_i \in \Sheaf(\CView{A, B}^{\op}) \).  On objects,
\[
  \tau_i(p) := \{ a \in \tau(p) \mid p_i \le p \textrm{ and } a_i = a \cdot f \textrm{ where \( f : p_i \to p \) } \}.
\]
Then \( \tau_i(p) \subseteq \sigma(p) \) for every \( i \) and \( p \).  For \( f : p \to p' \), we define \( \tau_i(f) \) as the restriction of \( \tau(f) : \tau(t') \to \tau(t) \) to \( \tau_i(t') \subseteq \tau(t') \).  It is easy to see that \( \tau_i \) is a functor.  Then we have
\[
  \tau \cong \tau_1 + \dots + \tau_n.
\]
To see this, consider \( a \in \tau(p) \) for some \( p \).  Let \( p' \) be the first two moves of \( p \) and let \( a' = a \cdot f \), where \( f : p' \to p \) (unique).  Then \( (p', a') \) is \( (p_i, a_i) \) for some \( i \le n \).  Hence \( a \in \tau_i(p) \).  Furthermore such \( i \) is unique by the construction.  So we have the claimed natural isomorphism.  Letting \( \sigma_i := \iota_* \tau_i \), we obtain the statement.
\end{proof}

\begin{theorem}[Full completeness]\label{thm:nondet-full-complete}
Let \( \Gamma \) be a type environment, \( \kappa \) be a type and \( \sigma \in \Sheaf(\CPlay{\sem{\Gamma}, \sem{\kappa}}) \). If \( \sigma \) is finite, there exists a term \( \Gamma \vdash M : \kappa \) such that \( \sigma \cong \sem{M} \).
\end{theorem}
\begin{proof}
By induction on the number of elements in \( \coprod_{p \in \CView{\sem{\Gamma}, \sem{\kappa}}} \sigma(p) \) and the structure of \( \kappa \).  If \( \kappa = \kappa_1 \to \kappa_2 \), consider \( \Lambda^{-1}(\sigma) \in \Sheaf(\CPlay{\sem{\Gamma} \times \sem{\kappa_1}, \sem{\kappa_2}}) \) and apply the induction hypothesis.  Suppose that \( \kappa = \T \).  If \( \sigma \) has several initial responses, then by applying Lemma~\ref{lem:nondet-unique-initial-decompose}, we have \( \sigma = \sigma_1 + \dots + \sigma_n \) (\( n \ge 2 \)).  By the induction hypothesis, we have \( \sigma_i \cong \sem{M_i} \) for every \( i \) and thus \( M_1 + \dots + M_n \) is the required term.  Suppose that \( \sigma \) is deterministic on initial response.  Let \( (m_1 m_2, a) \) be the unique response.  Since \( \sigma \in \Sheaf(\CPlay{\sem{\Gamma}, \sem{\T}}) \), \( m_1 \) is the unique initial move of \( \sem{\T} \) and \( m_2 \) be the unique initial move of \( \sem{\kappa_k} \), where \( x_k : \kappa_k \in \Gamma \) for some \( x_k \).  Suppose that \( \kappa_k = \kappa'_1 \to \dots \to \kappa'_l \to \T \).  We define the sheaf \( \tau' \in \Sheaf(\sem{\Gamma}, \sem{\kappa'_1} \times \dots \times \sem{\kappa'_l}) \) by:
\[
  \tau'(p) := \sigma(\chi_{(m_1m_2)}^{-1}(p)),
\]
where \( \chi_{(m_1m_2)}^{-1} : \CView{\sem{\Gamma}, \sem{\kappa'_1} \times \dots \times \sem{\kappa'_n}} \to \CView{\sem{\Gamma}, \sem{\T}} : p \mapsto m_1 m_2 p \) (see Section~\ref{sec:key-lemma}).  Then \( \iota^* \sigma \cong ((m_1m_2) \rhd \tau') \) because \( \sigma \) is deterministic on initial response.  So by Lemma~\ref{lem:key-lemma}, we have
\[
  \sigma \cong \langle \pi_k, \tau' \rangle; \mathbf{ev}
\]
where \( \pi_k : \sem{\Gamma} \to \sem{\kappa_k} \) is the projection.  By the induction hypothesis, we have \( M_i \) for each \( i \in [l] \) such that \( \tau' \cong \langle \sem{M_1}, \dots, \sem{M_l} \rangle \).  Recall that \( \sem{x_k} \cong \pi_k \).  Since \( \CGame \) is a CCC, the application of the product can be rewritten by the series of applications.  Hence \( \sigma \cong \sem{x_k\,M_1\,\dots\,M_l} \) as desired.
\end{proof}

\begin{example}
Let \( {\btt} = \lam xy. x \) and \( \bff = \lam xy.y \).  Recall the example in Introduction, \( M_1 = \lam f. (f\,\btt) + (f\,\bff) \) and \( M_2 = (\lam f. f\,\btt) + (\lam f. f\,\bff) \).  Then \( \iota^* \sem{M_1} = \tau_1 \) and \( \iota^* \sem{M_2} = \tau_2 \), where sheaves \( \tau_1 \) and \( \tau_2 \) over P-views can be found in Example~\ref{eg:view-dependency}.
\end{example}

\section{Sheaves model for probabilistic \( \lambda_{\to} \)}
We have seen that a term of the non-deterministic \( \lambda_{\to} \) is modelled by a sheaf \( \sigma \) which maps a play \( s \) to a (finite) set \( \sigma(s) \).  An element of \( \sigma(s) \) represents a particular choice of branches by which the term behaves like \( s \).

In this section, we shall study a non-deterministic sheaf \( \sigma \) equipped with a \emph{weight map} \( \mu \) which assigns each choice \( (s, a) \) (where \( s \) is a play and \( a \in \sigma(s) \)) with a positive real number \( \mu(s, a) \).

\subsection{The target calculus: weighted and probabilistic \( \lambda_{\to} \)}
The target language is an extension of the nondeterministic \( \lambda_{\to} \) studied in the previous section.  The new feature is the term constructor \( c \cdot M \), where \( c \) is a positive real number.  The additional equations are: \tk{Corrected.  2/July/2014.  The previous rule for the application is \( (c \cdot M) (c' \cdot N) = (c c') \cdot (M\,N) \)}
\begin{align*}
  \lam x. (c \cdot M) &= c \cdot (\lam x. M)
  &
  c \cdot (M + N) &= (c \cdot M) + (c \cdot N)
  \\
  c_1 \cdot (c_2 \cdot M) &= (c_1 c_2) \cdot M
  &
  (c \cdot M)\,N &= c \cdot (M\,N)
\end{align*}
and \( c \cdot \bot = \bot \).  These equations are admissible in the sense that \( M = N \) implies \( M \) and \( N \) are observably equivalent in the standard call-by-name operational semantics (where the observable is the probability of convergence).
The probabilistic \( \lambda_{\to} \) is a fragment of this calculus in which nondeterministic branch and the weight construct are restricted to the form \( (c_1 \cdot M_1) + \dots + (c_n \cdot M_n) \), where \( \sum_{i = 1}^n c_i \le 1 \).

\begin{remark}
The rule \( M\,(c \cdot N) = c \cdot (M\,N) \) is unsound, because the application is not linear on the argument.  For instance, if the argument is called twice as in \( (\lam f. f (f (z))) (c \cdot \lam x. x) \), the resulting coefficient is \( c^2 \):
\begin{align*}
  (\lam f. f (f (z))) (c \cdot \lam x. x)
  &= (c \cdot \lam x.x) ((c \cdot \lam x.x) z) \\
  &= c \cdot ((\lam x.x) (c \cdot ((\lam x.x) z)) \\
  &= c \cdot c \cdot ((\lam x.x) z) = c^2 z.
\end{align*}
Similarly, if the argument never be called as in \( (\lam x. z) (c \cdot N) \), the coefficient \( c \) does not affect, e.g.~\( (\lam x. z) (c \cdot N) = z = (\lam x.z)\,N \).  
\end{remark}

A \emph{normal form} is defined by:
\[
  R := c_1 \cdot Q_1 + \dots + c_n \cdot Q_n
\qquad
  Q := \lambda x_1 \dots x_k. y\,R_1\,\dots\,R_n,
\]
where \( y\,R_1 \dots R_n \) is fully applied.  Every term has a unique normal form (modulo commutation of the non-deterministic branches), or is equivalent to \( \lambda x_1 \dots x_k. \bot \).  Note that \( 2 \cdot M + 2 \cdot M \neq 4 \cdot M \).

\subsection{Sheaves with weight}

\begin{definition}[Weight]
Let \( F \) be a functor \( \mathbb{D}^{\op} \to \Set \).  A \emph{weight map} \( \mu \) assigns, for each \( s \in \mathbb{D} \) and \( a \in F(s) \), a positive real number \( \mu(s, a) \in \PosReal \).
\end{definition}

Let \( \sigma \in \Sheaf(\CPlay{A, B}) \) and \( \mu \) be a weight map.  Given a morphism \( f : s \to t \) in \( \CPlay{A, B} \) and an element \( a \in \sigma(t) \), we define \( \mu(f, a) := \mu(t, a) / \mu(s, a \cdot f) \).  Notice that \( \mu(g \circ f, a) = \mu(g, a) \mu(f, a \cdot g) \).

\begin{definition}[Innocence on weight]
Let \( \sigma \in \Sheaf(\CPlay{A, B}) \) be a sheaf and \( \mu \) be a weight map.  The weight map \( \mu \) is \emph{innocent} if it satisfies the following conditions: \((1)\) \( \mu(\varepsilon, \ast) = 1 \), and \((2)\) given a covering family \( \{ f : s \to u,\ g : t \to u \} \) and \( a \in \sigma(u) \), consider the pullback diagram
\[\xymatrix@R-0.1cm@C.8cm{
s \times_u t \ar[r] \ar[d]^{g^*(f)}
&
s \ar[d]^{f}
\\
t \ar[r]^{g}
&
u
}
\]
then \( \mu(f, a) = \mu(g^*(f),\; a \cdot g) \).
\end{definition}

The typical case is that \( u = v_0 v_1 v_2 \), \( s = v_0 v_1 \), \( t = v_0 v_2 \) and \( s \times_u t = v_0 \).  Intuitively \( \mu(f, a) \) is the weight of playing \( v_2 \) from \( s = v_0 v_1 \) (that reaches to the state \( a \in \sigma(s) \)) and \( \mu(g^*(f), a \cdot g) \) is the weight of playing \( v_2 \) from \( v_0 \) (that reaches to the state \( a \cdot g \in \sigma(t) \), the restriction of \( a \) to \( t \)).  The innocence of the weight map requires that the weight for playing \( v_2 \) is independent of the situation. 

\begin{definition}[Weighted innocent strategy]
A \emph{weighted innocent strategy} over pairs \( (A, B) \) of arenas is a pair \( (\sigma, \mu) \) of an innocent non-deterministic strategy \( \sigma \in \Sheaf(\CPlay{A, B}) \) and an innocent weight map \( \mu \) for \( \sigma \).
\end{definition}

Similar to the deterministic / non-deterministic cases, a weighted innocent strategy is determined by its restriction on views.
\begin{lemma}
Assume \( \sigma, \sigma' \in \Sheaf(\CPlay{A, B}) \) and a natural isomorphism \( \phi : \sigma \stackrel{\cong}{\to} \sigma' \).  Let \( \mu \) and \( \mu' \) are innocent weight maps for \( \sigma \) and \( \sigma' \), respectively.  If \( \mu(p, a) = \mu'(p, \phi(a)) \) for every P-view \( p \in \CView{A, B} \), then \( \mu(s, a) = \mu'(s, \phi(a)) \) for every play \( s \in \CPlay{A, B} \). 
\end{lemma}
\begin{proof}
By induction on the length of \( s \).  Let \( s = s_0 m_1 m_2 \) be a play and \( e \in \sigma(s) \).  If \( s \) is a P-view, the claim is just assumed.  Suppose that \( s \) is not a P-view.  We have a covering family \( \{ f : s_0 \to s,\ g : \PView{s} \to s \} \).  Since the pullback \( g^*(f) : p_0 \to \PView{s} \) is in \( \CView{A, B} \),
\begin{align*}
  \mu(f, a)
  &= \mu(g^*(f), e \cdot g) = \mu'(g^*(f), \phi(e \cdot g)) \\
  &= \mu'(g^*(f), \phi(e) \cdot g) = \mu'(f, \phi(e)).
\end{align*}
By the induction hypothesis, we have
\[
  \mu(s_0, e \cdot f) = \mu'(s_0, \phi(e \cdot f)) = \mu'(s_0, \phi(e) \cdot f).
\]
So we conclude
\begin{align*}
  \mu(s, e)
  &= \mu(f, e) \mu(s_0, e \cdot f) \\
  &= \mu'(f, \phi(e)) \mu'(s_0, \phi(e) \cdot f) = \mu'(s, \phi(e))
\end{align*}
as desired.
\end{proof}

\begin{lemma}\label{lem:weight-view-extension}
Let \( \tau \in \Sheaf(\CView{A, B}) \).  Every weight map \( \mu_0 \) for \( \tau \) can be extended to an innocent weight map for \( \iota_* \tau \).
\end{lemma}
\begin{proof}
Given a non-empty P-view \( p = p_0 m_1 m_2 \in \CView{A, B} \) and \( e \in \tau(p) \), we define \( \delta(e) := \mu_0(p_0 \to p, e) \) (if \( p_0 \neq \varepsilon \)) and \( \delta(e) := \mu_0(p, e) \) (if \( p_0 = \varepsilon \)).  We give a weight map \( \mu \) for \( \iota_* \sigma \).  Let \( s = m_1 m_2 \dots m_n \in \CPlay{A, B} \) and \( x \in \iota_* \tau \).  Then \( x \) is of the form \( e_2 e_4 \dots e_n \), where \( e_k \in \tau(\PView{m_1 \dots m_k}) \) for every even number \( k \le n \).  The weight for \( x = e_2 e_4 \dots e_n \) is defined by:
\[
  \mu(s, e_2 e_4 \dots e_n) := \delta(e_2) \delta(e_4) \dots \delta(e_n).
\]
It is easy to see that \( \mu \) is innocent.
\end{proof}

So one can define a weighted innocent strategy as a pair of a sheaf over P-views and a weight function for it.

\begin{definition}
The \emph{category of weighted innocent strategies} \( \CWGame \) has arenas as objects and weighted innocent strategies as morphisms.  Here \( (\sigma_1, \mu_1) \) and \( (\sigma_2, \mu_2) \) are identifies if there exists a natural isomorphism preserving weights.
A composition of weighted innocent strategies \( (\sigma, \mu) \) and \( (\sigma', \mu') \) is \( ((\sigma; \sigma'), \mu'') \), where for each \( s \) and
\(
  (u, e, e') \in (\sigma; \sigma')(s) = \coprod_{u: \pi(u)=s} \sigma(\proj{u}{A, B}) \times \sigma'(\proj{u}{B, C}),
\)
where \( e \in \sigma(\proj{u}{A, B}) \) and \( e' \in \sigma'(\proj{u'}{B, C}) \), we define
\[
  \mu''(s, (u, e, e')) = \mu(\proj{u}{A, B}, e) \, \mu'(\proj{u}{B, C}, e').
\]
Associativity of the composition can be easily shown.
\end{definition}

\begin{lemma}
\( \CWGame \) is a cartesian closed category.
\end{lemma}
\begin{proof}
Given a deterministic innocent strategy \( \sigma \in \Sheaf(\CPlay{A, B}) \), the trivial weight map \( \mu \) is defined by \( \mu(s, e) = 1 \) for every \( s \) and \( e \).  Then \( \ident_A \) with the trivial weight map is the identity and \( \pi_1 \in \Sheaf(\CPlay{A \times B, A}) \) and \( \pi_2 \in \Sheaf(\CPlay{A \times B, B}) \) with the trivial weight maps are projections.  The natural isomorphism \( \CGame(A \times B, C) = \Sheaf(\CPlay{A \times B, C}) \cong \Sheaf(\CPlay{A, B \Rightarrow C}) = \CGame(A, B \Rightarrow C) \) has obvious extension to weighted innocent strategies.  Hence \( \CWGame \) is a CCC.
\end{proof}

\subsection{Semantics of weighted \( \lambda_{\to} \)}
Let \( \tau \) and \( \tau' \) be sheaves over P-views of \( (A, B) \) and \( \mu_0 \) and \( \mu'_0 \) be weight maps for \( \tau \) and \( \tau' \), respectively.  The weight map \( [\mu_0, \mu_0'] \) for \( \tau + \tau' \) is defined by \( [\mu_0, \mu_0'](p, e) := \mu_0(p, e) \) (if \( e \in \tau(p) \)) and \( [\mu_0, \mu_0'](p, e) := \mu_0'(p, e) \) (if \( e \in \tau'(p) \)).  We define \( c \otimes \mu_0 \) by \( (c \otimes \mu_0)(p, e) := c \mu_0(p, e) \).

The same operations can be defined for weighted innocent strategies through Lemma~\ref{lem:weight-view-extension}.
Given a weighted innocent strategy \( (\sigma, \mu) \), we define \( c \otimes \mu \) the unique extension of \( c \otimes \mu_0 \) to \( \sigma \), where \( \mu_0 \) is the restriction of \( \mu \) to P-views.  Then \( (c \otimes \mu)(s, e) = c^k \mu(s, e) \), where \( k \) is the number of the moves in \( s \) that point to \( \star \).  It is easy to check that the equations about weights are sound for this interpretation, by using the next lemma.
\begin{lemma}
Let \( s \) be a well-opened play and \( e \in \sigma(s) \).  Then \( (c \otimes \mu)(s, e) = c (\mu(s, e)) \).
\end{lemma}

\begin{lemma}
\( M = N \) iff \( \sem{M} = \sem{N} \).
\end{lemma}

Let \( B \) be a prime arena.  A weighted innocent strategy \( (\sigma, \mu) \) of \( (A, B) \) is \emph{deterministic on initial response} if \( \coprod_{s \in \ImmExt{\T}} \sigma(s) \) is singleton and \( \mu(s, e) = 1 \) for its unique element \( (s, e) \).  The next lemma can be proved by the same way as Lemma~\ref{lem:nondet-unique-initial-decompose}.
\begin{lemma}\label{lem:weight-unique-initial-decompose}
Every finite weighted innocent \( (\sigma, \mu) \) strategy can be decomposed as \( c_1 \otimes (\sigma_1, \mu_1) + \dots + c_n \otimes (\sigma_n, \mu_n) \), where \( (\sigma_i, \mu_i) \) is deterministic on initial response.
\end{lemma}

The full completeness for the weighted calculus is proved by the same technique as in the proof of Theorem~\ref{thm:nondet-full-complete}, using Lemma~\ref{lem:weight-unique-initial-decompose}.
\begin{theorem}[Full completeness]
Let \( (\sigma, \mu) \) be a weighted innocent strategy for \( (\sem{\Gamma}, \sem{\kappa}) \) and suppose that \( \sigma \) is finite.  Then there exists a term \( \Gamma \vdash M : \kappa \) such that \( (\sigma, \mu) \cong \sem{M} \).
\end{theorem}

\subsection{Semantics of probabilistic \( \lambda_{\to} \)}
\tkchanged{A weighted innocent strategy \( (\sigma, \mu) \) is probabilistic if, for every odd-length play \( s = s_0 m \) and \( e_0 \in \sigma(s_0) \), the sum of weights of possible responses that extends \( (s, e_0) \) is less than \( 1 \).}

\begin{definition}
A weighted innocent strategy \( (\sigma, \mu) \) over \( (A, B) \) is \emph{probabilistic} if, for every odd-length play \( s = s_0 m \) and \( e_0 \in \sigma(s_0) \), we have
\[
  \sum_{t \in \ImmExt{s}} \sum_{e \in \sigma(t):\, e \cdot f_t = e_0} \mu(f_t, e) \le 1
\]
where \( f_t : s_0 \to t \) is the prefix embedding.  It can be strictly less than \( 1 \); the difference is the probability of divergence.  A sheaf \( \tau \) over views with a weight map \( \mu_0 \) is \emph{probabilistic} when the same condition holds (but \( s \) is restricted to P-views).
\end{definition}

\begin{lemma}
\( (\sigma, \mu) \) is probabilistic iff its restriction to views is.
\end{lemma}
\iflong 
\begin{proof}
Let \( \sigma \in \Sheaf(\CPlay{A, B}) \) and \( \tau = \iota^* \sigma \in \CView{A, B} \).  Let \( s = s_0 m \) be an odd-length play and \( e_0 \in \sigma(s_0) \).  We prove
\[
  \sum_{t \in \ImmExt{s}} \sum_{e \in \sigma(t):\, e \cdot f_t = e_0} \mu(f_t, e) \le 1
\]
by induction on the length \( s \), where \( f_t : s_0 \to t \) is the prefix embedding.  If \( \PView{s} = s \), then every \( t \in \ImmExt{s} \) is a P-view.  Hence the claim follows from the assumption.

Assume that \( \PView{s} \neq s \).  Let \( s_0 = m_1 \dots m_n \), \( m_j \) be the justifier of \( m \) and \( p_0 = \PView{m_1 \dots m_j} \).  Consider the covering family \( \{ f_t : s_0 \to t,\; g_t : \PView{t} \to t \} \) for every \( t \).  Then we have \( \mu(f_t, e) = \mu(g_t^*(f_t), e \cdot g_t) \) for every \( t \in \ImmExt{s} \).  So it suffices to prove that
\[
  \sum_{t \in \ImmExt{s}} \sum_{e \in \sigma(t):\, e \cdot f_t = e_0} \mu(g_t^*(f_t), e \cdot g_t) \le 1
\]
Since the P-view of \( t \in \ImmExt{s} \) is given by \( \PView{t} = \PView{s_0 m m'} = p_0 m m' \) (for some \( m' \)), we have a bijection from \( \ImmExt{s} \) to \( \ImmExt{p_0 m} \).  Since \( \{ f_t, g_t \} \) is a covering family, a pair \( (b_t, d_t) \in \sigma(s_0) \times \sigma(\PView{t}) \) such that \( b_t \cdot f_t^*(g_t) = d_t \cdot g_t^*(f_t) \) bijectively corresponds to \( e \in \sigma(t) \).  So there exists a bijection between \( \{ e \in \sigma(t) \mid e \cdot f_t = e_0 \} \) and \( \{ e \in \sigma(\PView{t}) \mid e \cdot g_t^*(f_t) = e_0 \cdot f_t^*(g_t) \} \).  Since \( g_t^*(f_t) : p_0 \to \PView{t} \) is the prefix embedding and \( f_t^*(g_t) : \PView{s_0} \to s_0 \) is the P-view embedding that is independent of \( t \), we conclude
\begin{align*}
  & \sum_{t \in \ImmExt{s}} \sum_{{e \in \sigma(t)} \atop {e \cdot f_t = e_0}} \mu(g_t^*(f_t), e \cdot g_t) \\
  =&
  \sum_{p \in \ImmExt{p_0 m}} \sum_{{e \in \sigma(p)} \atop {e \cdot (p \ge p_0) = e_0 \cdot h}} \mu((p \ge p_0), e) \le 1
\end{align*}
where \( h : \PView{s_0} \to s_0 \) is the P-view embedding.
\end{proof}
\fi 

Because the probabilistic \( \lambda_{\to} \) is a fragment of the weighted calculus, all the properties including soundness and adequacy are applicable fro the probabilistic calculus.  Full completeness can be proved by the same way as the weighted case.
\begin{theorem}[Full completeness]
Let \( \kappa \) be a simple type, \( \sigma \in \Sheaf(\CPlay{\sem{\Gamma}, \sem{\kappa}}) \) be finite and \( \mu \) be a probabilistic weight.  Then \( (\sigma, \mu) = \sem{M} \) for some probabilistic term \( \Gamma \vdash M : \kappa \). 
\end{theorem}

\paragraph{Concluding remarks}
As presented, our model treats neither  recursion nor primitive data types such as boolean. Further the target languages are restricted to simply-typed calculi. However we believe that these restrictions can be relaxed.

We will apply the sheaf-theoretic approach in the paper to study the model checking of non-determinsitic calculi, such as non-deterministic PCF and its call-by-value version, and to develop a semantics of \tkchanged{refinement} dependent types.

\bibliographystyle{abbrvnat}
\bibliography{library,tsukada}

\begin{thebibliography}{19}
\providecommand{\natexlab}[1]{#1}
\providecommand{\url}[1]{\texttt{#1}}
\expandafter\ifx\csname urlstyle\endcsname\relax
  \providecommand{\doi}[1]{doi: #1}\else
  \providecommand{\doi}{doi: \begingroup \urlstyle{rm}\Url}\fi

\bibitem[Abramsky and McCusker(1997)]{AbramskyM97}
S.~Abramsky and G.~McCusker.
\newblock Linearity, sharing and state: a fully abstract game semantics for
  {Idealized Algol} with active expressions.
\newblock In \emph{Algol-like Languages}, pages 297--329. Birkha\"{u}ser, 1997.

\bibitem[Abramsky et~al.(2000)Abramsky, Jagadeesan, and
  Malacaria]{AbramskyJM00}
S.~Abramsky, R.~Jagadeesan, and P.~Malacaria.
\newblock Full abstraction for pcf.
\newblock \emph{Inf. Comput.}, 163\penalty0 (2):\penalty0 409--470, 2000.

\bibitem[Beilinson(2012)]{Beilinson12}
A.~Beilinson.
\newblock P-adic periods and derived de {Rham} cohomology.
\newblock \emph{J.~AMS}, 25\penalty0 (3):\penalty0 715--738, 2012.

\bibitem[Castellan et~al.(2014)Castellan, Clairambault, and
  Winskel]{CastellanCW14}
S.~Castellan, P.~Clairambault, and G.~Winskel.
\newblock Concurrent {Hyland-Ong} games.
\newblock Lecture slides, IHP Workshop on Semantics of Proofs and Programs,
  2014.

\bibitem[Danos and Harmer(2002)]{DanosH02}
V.~Danos and R.~Harmer.
\newblock Probabilistic game semantics.
\newblock \emph{ACM Trans. Comput. Log.}, 3\penalty0 (3):\penalty0 359--382,
  2002.

\bibitem[Eberhart et~al.(2013)Eberhart, Hirschowitz, and Seiller]{EberhartHS13}
C.~Eberhart, T.~Hirschowitz, and T.~Seiller.
\newblock Fully abstract concurrent games for pi.
\newblock \emph{CoRR}, abs/1310.4306, 2013.

\bibitem[Harmer(1999)]{Harmer99}
R.~Harmer.
\newblock \emph{Games and Full Abstraction for Nondeterministic Languages}.
\newblock PhD thesis, Imperial College, 1999.

\bibitem[Harmer and McCusker(1999)]{HarmerM99}
R.~Harmer and G.~McCusker.
\newblock A fully abstract game semantics for finite nondeterminism.
\newblock In \emph{LICS}, pages 422--430, 1999.

\bibitem[Hirschowitz and Pous(2012)]{HirschowitzP12}
T.~Hirschowitz and D.~Pous.
\newblock Innocent strategies as presheaves and interactive equivalences for
  ccs.
\newblock \emph{Sci. Ann. Comp. Sci.}, 22\penalty0 (1):\penalty0 147--199,
  2012.

\bibitem[Hyland and Ong(2000)]{HylandO00}
J.~M.~E. Hyland and C.-H.~L. Ong.
\newblock On full abstraction for {PCF: I, II, and III}.
\newblock \emph{Inf. Comput.}, 163\penalty0 (2):\penalty0 285--408, 2000.

\bibitem[Jung et~al.(2008)Jung, Moshier, and Vickers]{JungMV08}
A.~Jung, M.~A. Moshier, and S.~J. Vickers.
\newblock Presenting dcpos and dcpo algebras.
\newblock \emph{Electr. Notes Theor. Comput. Sci.}, 218:\penalty0 209--229,
  2008.

\bibitem[Lane and Moerdijk(1992)]{MacLaneM92}
S.~M. Lane and I.~Moerdijk.
\newblock \emph{Sheaves in Geometry and Logic}.
\newblock Springer-Verlag, 1992.

\bibitem[Levy(2013)]{Levy13}
P.~Levy.
\newblock Morphisms between plays.
\newblock Lecture Slides, GaLoP, 2013.

\bibitem[Nickau(1994)]{Nickau94}
H.~Nickau.
\newblock Hereditarily sequential functionals.
\newblock In \emph{LFCS}, pages 253--264, 1994.

\bibitem[Ong(2006)]{Ong06}
C.-H.~L. Ong.
\newblock On model-checking trees generated by higher-order recursion schemes.
\newblock In \emph{LICS}, pages 81--90, 2006.

\bibitem[Rideau and Winskel(2011)]{RideauW11}
S.~Rideau and G.~Winskel.
\newblock Concurrent strategies.
\newblock In \emph{LICS}, pages 409--418, 2011.

\bibitem[Staton and Winskel(2010)]{StatonW10}
S.~Staton and G.~Winskel.
\newblock On the expressivity of symmetry in event structures.
\newblock In \emph{LICS}, pages 392--401, 2010.

\bibitem[Tsukada and Ong(2014)]{TsukadaO14}
T.~Tsukada and C.-H.~L. Ong.
\newblock Compositional higher-order model checking via $\omega$-regular games
  over {B\"ohm} trees".
\newblock In \emph{CSL/LICS}, 2014.

\bibitem[Verdier(1972)]{Verdier72}
J.-L. Verdier.
\newblock Fonctorialit\'e de cat\'{e}gories de faisceaux.
\newblock In \emph{Th\'{e}orie des topos et cohomologie \'{e}tale de
  sch\'{e}mas (SGA 4), Tome 1}, pages 265--298. Springer-Verlag, 1972.
\newblock Lect. Notes in Math. 269.

\end{thebibliography}

\end{document}